\documentclass[preprint,3p]{elsarticle}

\usepackage[english]{babel}
\usepackage{amsfonts}
\usepackage{amsmath}
\usepackage{amssymb}                                                                                                                
\usepackage{amsthm}
\usepackage[utf8]{inputenc}                                    
\usepackage[T1]{fontenc}
\usepackage{mathrsfs}
\usepackage[small,bf]{caption}
\usepackage{subfig}
\usepackage{stmaryrd}
\usepackage{float}
\usepackage{changepage}
\usepackage{enumitem}

\journal{arXiv}

\DeclareMathOperator*{\mnm}{min}
\DeclareMathOperator*{\mxm}{max}

\newcommand{\eps}{\varepsilon}
\DeclareMathOperator*{\supp}{supp}
\mathchardef\hyp="2D

\makeatletter
\DeclareFontFamily{OMX}{MnSymbolE}{}
\DeclareSymbolFont{MnLargeSymbols}{OMX}{MnSymbolE}{m}{n}
\SetSymbolFont{MnLargeSymbols}{bold}{OMX}{MnSymbolE}{b}{n}
\DeclareFontShape{OMX}{MnSymbolE}{m}{n}{
    <-6>  MnSymbolE5
   <6-7>  MnSymbolE6
   <7-8>  MnSymbolE7
   <8-9>  MnSymbolE8
   <9-10> MnSymbolE9
  <10-12> MnSymbolE10
  <12->   MnSymbolE12
}{}
\DeclareFontShape{OMX}{MnSymbolE}{b}{n}{
    <-6>  MnSymbolE-Bold5
   <6-7>  MnSymbolE-Bold6
   <7-8>  MnSymbolE-Bold7
   <8-9>  MnSymbolE-Bold8
   <9-10> MnSymbolE-Bold9
  <10-12> MnSymbolE-Bold10
  <12->   MnSymbolE-Bold12
}{}

\let\llangle\@undefined
\let\rrangle\@undefined
\DeclareMathDelimiter{\llangle}{\mathopen}%
                     {MnLargeSymbols}{'164}{MnLargeSymbols}{'164}
\DeclareMathDelimiter{\rrangle}{\mathclose}%
                     {MnLargeSymbols}{'171}{MnLargeSymbols}{'171}
\makeatother

\newtheorem{theorem}{Theorem}[section]
\newtheorem{lemma}[theorem]{Lemma}
\newtheorem{proposition}[theorem]{Proposition}

\newdefinition{definition}[theorem]{Definition}
\newdefinition{notation}[theorem]{Notation}
\newdefinition{remark}[theorem]{Remark}
\newdefinition{example}[theorem]{Example}

\begin{document}
\begin{frontmatter}

\title{Weighted Automata and Logics Meet Computational Complexity}

\author{Peter Kostol\'anyi}
\ead{kostolanyi@fmph.uniba.sk}
\address{Department of Computer Science, Comenius University in Bratislava, \\
Mlynsk\'a dolina, 842 48 Bratislava, Slovakia}

\begin{abstract}
Complexity classes such as $\#\mathbf{P}$, $\oplus\mathbf{P}$, $\mathbf{GapP}$, $\mathbf{OptP}$, $\mathbf{NPMV}$, or the class of fuzzy languages realised 
by~polynomial-time fuzzy nondeterministic Turing machines, can all be described in~terms of~a~class~$\mathbf{NP}[S]$ for~a~suitable semiring $S$, defined via weighted Turing machines over $S$
similarly as $\mathbf{NP}$ is defined via the~classical nondeterministic Turing machines. Other complexity classes of decision problems can be lifted to the quantitative world
using the same recipe as well, and the resulting classes relate to the original ones in the same way as weighted automata or logics relate to their unweighted counterparts.
The article surveys these too-little-known connexions between weighted automata theory and computational complexity theory implicit in the existing literature,    
suggests a systematic approach to the study of weighted complexity classes, and presents several new observations strengthening the relation between both fields. 
In particular, it is proved that a natural extension of the Boolean satisfiability problem to weighted propositional logic is complete for the class $\mathbf{NP}[S]$
when $S$ is a finitely generated semiring. Moreover, a class of semiring-valued functions~$\mathbf{FP}[S]$ is introduced for each semiring $S$ as a~counterpart to the class $\mathbf{P}$,
and the relations between $\mathbf{FP}[S]$ and $\mathbf{NP}[S]$ are considered. 
\end{abstract}

\end{frontmatter}

\section{Introduction}

\emph{Weighted automata theory} \cite{berstel2011a,droste2009a,droste2021a,kuich1997a,kuich1986a,sakarovitch2009a,salomaa1978a}, origins of which go back to the seminal article of~M.-P.~Sch\"utzenberger~\cite{schutzenberger1961a}, encompasses 
the study of \emph{quantitative generalisations} of automata and related formalisms such as grammars \cite{droste2014a,inoue2023a,stanat1972a} and systems of equations \cite{kuich1986a,petre2009a}, rational expressions \cite{sakarovitch2009a,sakarovitch2021a}, and MSO logics \cite{droste2007a,droste2009c}.
A weight taken from some structure -- most typically a \emph{semiring} -- is assigned to each transition of a weighted automaton, so that these automata no longer describe formal languages;
instead, they realise functions mapping words over some alphabet to elements of the underlying semiring. After identifying the most natural algebra for such functions, one usually interprets
them as \emph{formal power series} in several noncommutative variables \cite{berstel2011a,droste2009b,droste2021a,sakarovitch2009a,sakarovitch2009b,salomaa1978a}, also called \emph{weighted languages} \cite{herrmann2019a,maletti2017a}.  
  
Among the objects of study of weighted automata theory, weighted finite automata and the corresponding class of rational series are undoubtedly the best understood~\cite{sakarovitch2009b}.
Nevertheless, a great deal of research has~also focused on weighted generalisations of models that are not expressive equivalents of finite automata: for~instance, there are now well-developed theories of weighted counterparts to context-free languages \cite{kuich1986a,petre2009a},
one-counter languages \cite{kuich1983a}, or Lindenmayer systems \cite{honkala2009a}.  
  
On the other hand, \emph{weighted Turing machines} received comparably little attention. These were introduced as ``algebraic Turing machines'' 
by C.~Damm, M.~Holzer, and~P.~McKenzie~\cite{damm2002a}, who also discovered many of the applications of such machines to computational complexity surveyed in this article.
However, the authors used this model mainly as an auxiliary tool for studying complexity aspects of~evaluating tensor formulae over various semirings, and they did not make a~natural connexion to weighted automata theory explicit.
Their work thus remained essentially unnoticed by the weighted automata community.\goodbreak      

To the author's best knowledge, an explicit study of weighted Turing machines over structures incorporating at least all semirings only appeared in~\cite{li2015},
where several variants of such machines over strong bimonoids\footnote{The theory of weighted \emph{finite automata} over strong bimonoids is a well-understood generalisation of their usual theory over semirings \cite{ciric2010a,droste2010a}; strong bimonoids are essentially semirings without distributivity.} were considered.   
There also is a mention of weighted Turing machines in~\cite{burgin2008a}; however, these are understood there on~a~more-or-less informal level
and in a much broader sense. Finally, a class of so-called ``semiring Turing machines'' similar to weighted Turing machines was recently introduced in~\cite{eiter2021a,eiter2023a}.      
\smallskip

This article can be seen as an attempt to convey the significance of \emph{weighted Turing machines}, defined by analogy to weighted automata, for 
a substantial part of computational complexity theory.

As decision problems are captured by languages, the usual complexity classes of decision problems -- such as $\mathbf{P}$, $\mathbf{NP}$, $\mathbf{PSPACE}$, or $\mathbf{EXPTIME}$ -- are
actually classes of languages. On the other hand, counting or~optimisation problems no longer correspond to languages, but to their quantitative extensions.
Applying an approach standard in weighted automata theory, these problems can be regarded as \emph{formal power series}. The~classes such as $\#\mathbf{P}$ or $\mathbf{OptP}$
can thus be seen as weighted -- or quantitative -- complexity classes of~noncommutative formal power series.

The key observation to be highlighted in this article is that many of such quantitative complexity classes, commonly studied in computational complexity theory, can in fact be naturally characterised by weighted Turing machines
over a suitable semiring. For instance, one may consider a complexity class $\mathbf{NP}[S]$ consisting of all power series realised by~polynomial-time
weighted Turing machines over a semiring $S$. As usual, setting~$S$ to the Boolean semiring corresponds to disregarding weights -- that is, $\mathbf{NP}[S]$ becomes just $\mathbf{NP}$ for $S = \mathbb{B}$.
On the other hand, the semiring of natural numbers~$\mathbb{N}$ typically takes the role of counting;
it is thus not surprising that the class $\mathbf{NP}[\mathbb{N}]$ coincides with the counting complexity class $\#\mathbf{P}$ \cite{valiant1979a,valiant1979b,arora2009a,goldreich2008a,papadimitriou1994a}.
This example, which can be seen as a leading source of inspiration for studying the classes $\mathbf{NP}[S]$,
was actually already observed -- and presented in their terminology of algebraic Turing machines -- by C.~Damm, M.~Holzer, and~P.~McKenzie~\cite{damm2002a}.

Moreover, several other examples of semirings $S$, for which $\mathbf{NP}[S]$ corresponds to a known complexity class, were identified in~\cite{damm2002a} as well:
over the finite field $\mathbb{F}_2 = \mathbb{Z}/2\mathbb{Z}$, the ``parity-$\mathbf{P}$'' class $\oplus\mathbf{P}$ \cite{papadimitriou1983a} is represented by $\mathbf{NP}[\mathbb{F}_2]$;
similarly, the class of~supports of~series from~$\mathbf{NP}[\mathbb{Z}/k\mathbb{Z}]$, for a natural number $k \geq 2$, corresponds to~$\mathbf{MOD_{k}P}$~\cite{beigel1992a,cai1990a},
and~the~class $\mathbf{NP}[\mathbb{Z}]$ can be seen as $\mathbf{GapP}$~\cite{gupta1991a,gupta1995a,fenner1994a}.     

This list is extended in this article by a few less straightforward examples: we observe that the class of~optimisation problems $\mathbf{OptP}[O(\log n)]$ \cite{krentel1988a}
can be captured by $\mathbf{NP}[\mathbb{N}_{\mxm}]$, where $\mathbb{N}_{\mxm}$ is the max-plus semiring of natural numbers; moreover, the class $\mathbf{OptP}$ \cite{krentel1988a}
can also be modelled as $\mathbf{NP}[S_{\mxm}]$ for $S_{\mxm}$ being a suitable ``max-plus semiring of binary words''.
Furthermore, over semirings of finite languages \smash{$2^{\Sigma^*}_{\mathit{fin}}$}, the classes \smash{$\mathbf{NP}[2^{\Sigma^*}_{\mathit{fin}}]$} relate to the complexity class of multivalued functions $\mathbf{NPMV}$ \cite{book1985a},
and the class of~all fuzzy languages realisable by polynomial-time fuzzy~Turing machines, understood in the sense of~\cite{wiedermann2004a} (see~also~\cite{bedregal2008a} for more on this model and \cite{zadeh1968a} for the~first mention of fuzzy algorithms),
can also be viewed as~$\mathbf{NP}[F_{*}]$ for a suitable semiring $F_{*}$ depending on the triangular norm $*$ considered.               

The existence of a problem complete for $\mathbf{NP}[S]$ when~$S$ is finitely generated as an~additive monoid was also essentially established in~\cite{damm2002a} and~the~subsequent article by M.~Beaudry and M.~Holzer~\cite{beaudry2007a} --
it~was~proved there that the evaluation problem for~scalar tensor formulae over~$S$ has this property.
In this article, we gain some more insight into the concept of $\mathbf{NP}[S]$-completeness by showing that in~fact already the~Cook-Levin theorem~\cite{sipser2013a} generalises to the weighted setting:
a suitably defined \emph{``satisfiability'' problem for weighted propositional logics}, $\mathsf{SAT}[S]$, is $\mathbf{NP}[S]$-complete with respect to polynomial-time many-one reductions whenever the semiring $S$ is finitely generated (in~the~usual algebraic sense, \emph{i.e.}, as~a~semiring).
Here, the~weighted propositional logics are precisely the~propositional fragments of~the~weighted MSO logics of~M.~Droste and~P.~Gastin~\cite{droste2007a,droste2009c}.
In~addition, we identify one more artificial \mbox{$\mathbf{NP}[S]$-complete} problem for~every finitely generated semiring $S$, inspired by~the~$\mathsf{TMSAT}$ problem of~S.~Arora and~B.~Barak~\cite{arora2009a};
this~problem involving weighted Turing machines can be proved to be $\mathbf{NP}[S]$-complete quite readily, but otherwise is of~limited use.
Finally, we observe that $\mathbf{NP}[S]$-complete problems with respect to polynomial-time many-one reductions do not exist when $S$ is not finitely generated.\goodbreak

In addition to the class $\mathbf{NP}[S]$, we also introduce its ``deterministic counterpart'' $\mathbf{FP}[S]$; this class of~semiring-valued functions generalises both
$\mathbf{P}$ and the class $\mathbf{FP}$, which relates to $\#\mathbf{P}$ similarly as~$\mathbf{P}$ does to~$\mathbf{NP}$~\cite{arora2009a}.
Quite naturally, $\mathbf{FP}[S] \subseteq \mathbf{NP}[S]$ holds for all semirings $S$; we also observe that there is a~semiring, for which this inclusion is provably strict.      
We finally gather some basic observations concerning the role of the semiring $S$ when it comes to the relation of the classes $\mathbf{FP}[S]$ and $\mathbf{NP}[S]$:
we prove that when $T$ is a factor semiring of $S$, then $\mathbf{FP}[S] = \mathbf{NP}[S]$ implies $\mathbf{FP}[T] = \mathbf{NP}[T]$.

The original among the results presented in this survey are based on~an~unpublished work of~the author from~2019. Meanwhile, some similar results were also obtained
by~T.~Eiter and~R.~Kiesel~\cite{eiter2021a,eiter2023a} in~their related framework of semiring Turing machines. 
\smallskip    

The author believes that a consistent study of weighted complexity classes and related concepts, suggested by this article, can~be worthwhile for at least the following three reasons:
\begin{enumerate}
\item{Traditionally, weighted automata theory has mostly dealt with what S.~Eilenberg~\cite{eilenberg1974a} aptly described as~rational and~algebraic phenomena -- or with their further specialisations.
Solid generalisations to~the~weighted setting have thus been obtained over the years for the classical theories of rational and context-free languages, corresponding to the two lowermost levels of the Chomsky hierarchy.
The~study of~weighted computational complexity might help to understand how the principal concepts of weighted automata theory relate to classes of languages beyond these two levels, usually studied within the theory of computation. 
It should become clearer that virtually all of the main concepts covered in a~typical first course on~formal languages and~automata theory such as~\cite{hopcroft1979a} -- including the~theory of~computation -- can not only be consistently generalised to the~weighted setting, but these
generalisations are actually meaningful and~often related to~important concepts from other areas.     
}
\item{Plenty of complexity classes have been introduced over the years, and many of them are quantitative in~their essence -- \emph{i.e.}, instead of~decision problems,
they classify problems in~which a~value from some algebra is assigned to each input \cite{arora2009a,papadimitriou1994a}. 
It turns out that these classes can often be characterised via weighted Turing machines, which makes the analogies with the corresponding complexity classes of~decision problems 
more transparent, while the nature of such analogies can be described by~the~semiring considered. The landscape of quantitative complexity classes thus in a sense becomes more readable.              
}
\item{Universal importance of weighted logics becomes apparent in the study of weighted complexity classes.
Weighted MSO logics over words were introduced by~M.~Droste and~P.~Gastin~\cite{droste2007a} in order to characterise the~class of~rational series in~a~similar spirit as~rational languages
are captured by~unweighted MSO logics over words~\cite{buchi1960a}; this result was also extended from finite words to other settings such as infinite words~\cite{droste2007b}, trees~\cite{droste2006a,droste2011b}, or~graphs~\cite{droste2015a}.
Later, weighted first-order logics and their relation to aperiodic weighted automata were also considered \cite{droste2019a}. 
In connexion to weighted complexity classes, weighted logics tend to arise in new unexpected ways. As we observe,
the Boolean satisfiability problem~$\mathsf{SAT}$ admits a natural weighted generalisation $\mathsf{SAT}[S]$ over a~semiring~$S$, which can be seen as the~``satisfiability'' problem for~the~\emph{weighted propositional logic} over $S$,
the propositional fragment of~the~weighted MSO logic over $S$. This problem turns out to be $\mathbf{NP}[S]$-complete whenever $S$ is finitely generated, which generalises the Cook-Levin theorem to the weighted setting.
In yet another direction, there have been attempts \cite{arenas2017a,arenas2020a,durand2021a,saluja1995a} to extend the~results of~\emph{descriptive complexity}~\cite{immerman1999a,gradel2007a} to~counting complexity classes.
A~relatively recent approach of M.~Arenas, M.~Mu\~noz, and~C.~Riveros~\cite{arenas2017a,arenas2020a} successfully uses weighted logics over the~semiring of~natural numbers for~this purpose.
The theory of~weighted complexity classes over abstract semirings opens a~door to~possible extensions of~these results, in~which quantitative descriptive complexity theory based on~weighted logics          
would be developed over some fairly general class of semirings. 

\emph{Note added in revision}: After the~first version of~this preprint was made available, the~study of~descriptive complexity in~the~weighted setting was initiated by~G.~Badia, M.~Droste, C.~Noguera, and~E.~Paul~\cite{badia2024a}.          
}
\end{enumerate}

\section{Preliminaries}

We denote by $\mathbb{N}$ the set of all natural numbers \emph{including} zero. Given a~set $X$, we write $2^X$ for~the~powerset of~$X$ and~$2_{\mathit{fin}}^X$ for the set of all finite subsets of $X$.
When not stated otherwise, alphabets are understood to~be finite and nonempty. The reversal of a word $w \in \Sigma^*$ over an alphabet $\Sigma$ is denoted by $w^R$.

A \emph{monoid} is a triple $M = (M,\cdot,1)$, where $M$ is a set, $\cdot$ is an associative binary operation on~$M$, and $1$ is a neutral element of $M$ with respect to this operation.
A monoid $(M,\cdot,1)$ is said to be \emph{commutative} if~$\cdot$~is. A~\emph{semiring} is~an~algebra $S = (S,+,\cdot,0,1)$ such that $(S,+,0)$ is a commutative monoid,
$(S,\cdot,1)$ is a~monoid, the multiplicative operation $\cdot$ distributes over $+$ both from left and~from right, and~$0 \cdot a = a \cdot 0 = 0$ holds for~all $a \in S$. 
See \cite{droste2009b,golan1999a,hebisch1998a} for a~reference on~semirings.  

A \emph{subsemiring} of a semiring $S$ is a semiring $T$ such that $T \subseteq S$, the operations of $T$ are those of~$S$ restricted to $T$, and the neutral elements of $T$ are the same as for $S$.
The \emph{subsemiring of $S$ generated by~a~set~$G \subseteq S$} is the smallest subsemiring $\langle G\rangle$ of~$S$ containing $G$ -- this is the same as the smallest superset of~$G \cup \{0,1\}$ contained in~$S$ and~closed under both operations of~$S$,
or the~intersection of all subsemirings of~$S$ containing~$G$.
A semiring $S$ is \emph{finitely generated} if it is generated by some finite $G \subseteq S$.
\smallskip  

Let $S$ be a semiring and $\Sigma$ an alphabet. A \emph{formal power series} over $S$ and $\Sigma$ \cite{berstel2011a,droste2009b,droste2021a,kuich1997a,sakarovitch2009a,sakarovitch2009b,salomaa1978a} is~a~mapping $r\colon\Sigma^* \to S$.
One usually writes $(r,w)$ for the value of $r$ on $w \in \Sigma^*$, and calls $(r,w)$ the~\emph{coefficient} of~$w$~in 
\begin{displaymath}
r = \sum_{w \in \Sigma^*} (r,w)\,w.
\end{displaymath}
The set of all series over $S$ and $\Sigma$ is denoted by $S\llangle\Sigma^*\rrangle$.

The \emph{support} of a~series $r \in S\llangle\Sigma^*\rrangle$ is the language $\supp(r)$ of all $w \in \Sigma^*$ such that $(r,w) \neq 0$.
A~series with finite support is called a \emph{polynomial} and the set of all polynomials over $S$ and $\Sigma$ is denoted by $S\langle\Sigma^*\rangle$.

The \emph{sum} of series $r_1,r_2 \in S\llangle\Sigma^*\rrangle$ is a series $r_1 + r_2$ defined for all $w \in \Sigma^*$ by
$(r_1 + r_2,w) = (r_1,w) + (r_2,w)$, and the \emph{Cauchy product} of $r_1,r_2 \in S\llangle\Sigma^*\rrangle$ is a series $r_1 \cdot r_2$ such that 
\begin{displaymath}
(r_1 \cdot r_2,w) = \sum_{\substack{u,v \in \Sigma^* \\ uv = w}} (r_1,u)(r_2,v)
\end{displaymath}
for all $w \in \Sigma^*$. This choice of a multiplicative operation is actually the reason behind adopting the terminology and notation of formal power series
for the mappings $r\colon\Sigma^* \to S$.  

A formal power series $r$ such that $(r,\eps) = a$ for some $a \in S$ and $(r,w) = 0$ for all $w \in \Sigma^+$ is identified with~$a$; similarly,
a series $r$ such that $(r,w) = 1$ for some $w \in \Sigma^*$ and $(r,x) = 0$ for all $x \in \Sigma^* \setminus \{w\}$ is~identified with~$w$. 

Given a semiring $S$ and alphabet~$\Sigma$, the algebras $(S\llangle\Sigma^*\rrangle,+,\cdot,0,1)$ and $(S\langle\Sigma^*\rangle,+,\cdot,0,1)$ are semi\-rings as~well~\cite{droste2009b}.
The~semirings of power series $S\llangle\Sigma^*\rrangle$ are common generalisations of both the usual semirings of~univariate formal power series and~the~semirings of~formal languages.
The~former are obtained when the~alphabet~$\Sigma$ is unary. The~latter arise when $S$ is the Boolean semiring $\mathbb{B} = (\mathbb{B},\lor,\land,0,1)$:
a~series $r$ over $\mathbb{B}$ can be identified with the~language $\supp(r)$; conversely, every language can be turned into a~series over~the~semiring~$\mathbb{B}$ by~taking its characteristic function.

A family of series $(r_i \mid i \in I)$ from $S\llangle\Sigma^*\rrangle$ is \emph{locally finite} if the~set $I(w) := \{i \in I \mid (r_i,w) \neq 0\}$ is finite for all $w \in \Sigma^*$. One then writes
$\sum_{i \in I} r_i = r$ for a series $r \in S\llangle\Sigma^*\rrangle$ defined by $(r,w) = \sum_{i \in I(w)}(r_i,w)$ for~all~$w \in \Sigma^*$.
\smallskip 

By an \emph{algebra of terms} $T(G)$ over a generator set $G$, we understand the algebra $(T(G),+,\cdot,0,1)$ of~terms of type $(2,2,0,0)$; see, e.g.,~\cite{almeida1994a}.
In other words, $T(G)$ is a~language over $G \cup \{0,1,+,\cdot,(,)\}$ such that each $g \in G \cup \{0,1\}$ is in $T(G)$, the~terms
$(t_1 + t_2)$ and $(t_1 \cdot t_2)$ are in $T(G)$ for all $t_1,t_2 \in T(G)$, and nothing else is in $T(G)$. The~operations $+$ and~$\cdot$ are defined by $+\colon (t_1,t_2) \mapsto (t_1 + t_2)$ and $\cdot\colon (t_1,t_2) \mapsto (t_1 \cdot t_2)$ for all $t_1,t_2 \in T(G)$. 
For every semiring $S$ generated by $G$, there is a~unique algebra homomorphism $h_G[S]\colon T(G) \to S$ such that $h_G[S](g) = g$ for all $g \in G \cup \{0,1\}$,
which we call the~\emph{evaluation homomorphism}; a term $t \in T(G)$ \emph{evaluates} to $a$ in $S$ if $h_G[S](t) = a$.
We often omit parentheses in~terms that have no effect on evaluation in~semirings -- e.g., we write $a+b+c$ instead of $(a + (b + c))$.     

Finally, recall that similarly as in any other variety of algebras \cite{almeida1994a,bergman2011a}, a semiring~$T$ is (isomorphic~to) a~\emph{factor semiring} of a semiring $S$ if and only if $T$ is a homomorphic image of $S$.\goodbreak

\section{Weighted Turing Machines}

We now define \emph{weighted Turing machines}, the basic model for our later considerations, related to~nondeterministic Turing machines in~the~same way as weighted finite
automata relate to nondeterministic finite automata without weights.
A~weight from a~semiring is thus assigned to~each transition of~a~weighted \mbox{Turing} machine; the~value of~a~computation is then obtained by~taking the~product 
of~its constituent transition weights, and~the~value of~an~input word~$w$ is the~sum of~values of~all computations~on~$w$.  
Weighted Turing machines were introduced as ``algebraic Turing machines'' by C.~Damm, M.~Holzer, and~P.~McKenzie~\cite{damm2002a}; the~change in~terminology reflects
the~natural place of these machines in~weighted automata theory.  

To guarantee validity of the above-described definition of weighted Turing machine semantics, we confine ourselves to machines with finitely many computations upon each input word. 
It is clear that nothing important is lost by such a~restriction when it comes to~complexity questions.
Moreover, we only consider single-tape weighted Turing machines for simplicity, as we eventually embark upon the~study of~properties that
do not depend on the number of tapes. Nevertheless, multitape weighted Turing machines can be defined by analogy.    

\begin{definition}
Let $S$ be a semiring and $\Sigma$ an alphabet. A~\emph{weighted Turing machine} over $S$ and input alphabet $\Sigma$ is
a~septuple $\mathcal{M} = (Q,\Gamma,\Delta,\sigma,q_0,F,\square)$, where $Q$ is a~finite set of states,
$\Gamma \supseteq \Sigma$ is a~working alphabet, $\Delta \subseteq (Q \setminus F) \times \Gamma \times Q \times (\Gamma \setminus \{\square\}) \times \{-1,0,1\}$
is a~set of transitions, $\sigma\colon \Delta \to S \setminus \{0\}$ is~a~transition weighting function, $q_0 \in Q$ is the initial state, $F \subseteq Q$ is the set of accepting states, and $\square \in \Gamma \setminus \Sigma$ is~the~``blank'' symbol.
\end{definition}

A transition $(p,c,q,d,s) \in \Delta$ has the following interpretation: if $\mathcal{M}$ finds itself in some state $p$, while a~letter~$c$ is being read 
by the machine's head, then $\mathcal{M}$ can perform a~computation step, in~which the state is changed to $q$, the symbol read by the head is rewritten to $d$, and finally the head moves $s$ cells to~the~right.
Note that there are no transitions leading from accepting states.

A \emph{configuration} of $\mathcal{M}$ can be defined in the same way as for nondeterministic Turing machines: it~is
a~unique description of~the~machine's state, the~contents of~the~working tape, as well as the position of~the~machine's head (any of the usual formal definitions is applicable).
Moreover, if $e = (p,c,q,d,s) \in \Delta$ is a~transition and $C_1,C_2$ are configurations of $\mathcal{M}$, we write 
$C_1 \rightarrow_e C_2$ if $C_1$ is a configuration with state $p$ and~the~head reading $c$, while~$C_2$ is obtained from $C_1$ by changing the~state to $q$,
rewriting the originally read $c$ to $d$, and moving the head $s$ cells to the right. We write $C_1 \rightarrow C_2$ if $C_1 \rightarrow_e C_2$
for some $e \in \Delta$.

Moreover, let us define a \emph{computation} of $\mathcal{M}$ to be a~finite sequence
$\gamma = (C_0, e_1, C_1, e_2, C_2, \ldots, C_{t - 1}, e_t, C_{t})$
such that $t \in \mathbb{N}$, $C_0,\ldots,C_t$ are configurations of~$\mathcal{M}$, $e_1,\ldots,e_t \in \Delta$ are transitions of~$\mathcal{M}$, $C_{k-1} \rightarrow_{e_k} C_{k}$ for $k = 1,\ldots,t$, 
and $C_0$ is a~configuration with the~initial state~$q_0$, a~word from $\Sigma^*$ on~the~working tape, and~the~head at~the~leftmost non-blank cell (if there is some).
We then write $\lvert\gamma\rvert := t$ to denote the~\emph{length} of~$\gamma$ and~$\sigma(\gamma) := \sigma(e_1)\sigma(e_2)\ldots\sigma(e_t)$ to denote the~\emph{value} of~$\gamma$.
We say that $\gamma$ is a~computation on~$w \in \Sigma^*$, and~write $\lambda(\gamma) = w$, if $C_0$ is a~configuration with $w$ on~the~working tape. 
We call $\gamma$ \emph{accepting} if $C_{t}$ is a~configuration with a~state from~$F$. We denote the set of~all computations of~$\mathcal{M}$ by~$C(\mathcal{M})$,
and~the~set of all accepting computations by $A(\mathcal{M})$.
Note that in general, computations might be lengthened by~going through some additional transitions -- but this is never the~case for~accepting computations.     

Now, the behaviour of a weighted Turing machine should be given by the sum of~the~monomials $\sigma(\gamma)\,\lambda(\gamma)$ over all accepting computations $\gamma$.
As this sum is infinite in general, the~behaviour is not well-defined for~all Turing machines over all semirings in~this way. For this reason, we confine ourselves
to what we call \emph{halting} weighted Turing machines in what follows: we say that a~weighted Turing machine~$\mathcal{M}$ over~a~semi\-ring~$S$~and~input alphabet~$\Sigma$ is~\emph{halting} if
the~set $C_w(\mathcal{M}) := \{\gamma \in C(\mathcal{M}) \mid \lambda(\gamma) = w\}$ is finite for all $w$ in $\Sigma^*$. The sum of $\sigma(\gamma)\,\lambda(\gamma)$
over all $\gamma \in A(\mathcal{M})$ is clearly locally finite -- and hence well-defined -- for such machines.
The~terminology of~``halting'' machines comes from the observation that finiteness of $C_w(\mathcal{M})$ is equivalent to~the~nonexistence of~an~infinite sequence $(C_0, e_1, C_1, e_2, C_2, \ldots)$
such that $(C_0, e_1, C_1, e_2, C_2, \ldots, C_{t - 1}, e_t, C_{t})$ is a~computation of~$\mathcal{M}$ on~$w$ for all $t \in \mathbb{N}$, by K\"onig's infinity lemma.\goodbreak      

\emph{All weighted Turing machines are assumed to be halting in what follows}. It is not hard to see that this is no real restriction when it comes to questions of computational complexity.   

\begin{definition}
Let $S$ be a semiring, $\Sigma$ an alphabet, and $\mathcal{M}$ a (halting) weighted Turing machine over $S$ and $\Sigma$. The \emph{behaviour} of $\mathcal{M}$ is
a formal power series $\|\mathcal{M}\| \in S\llangle\Sigma^*\rrangle$ defined by 
\begin{displaymath}
\|\mathcal{M}\| := \sum_{\gamma \in A(\mathcal{M})} \sigma(\gamma)\,\lambda(\gamma),
\end{displaymath}
the sum being over a locally finite family of monomials.
\end{definition}

One could define weighted Turing machines over algebras more general than semirings -- for instance, over strong bimonoids~\cite{ciric2010a,droste2010a} -- in~a~similar way.
Nevertheless, we stick with semirings here, as they provide the~arguably most robust and~best explored mathematical framework for~weighted automata theory.
Possible extensions are left for future research.

\section{$\mathbf{NP}[S]$: Complexity Classes of Power Series}

Let $\mathcal{M} = (Q,\Gamma,\Delta,\sigma,q_0,F,\square)$ be a weighted Turing machine over $S$ and~$\Sigma$.
Given a~word $w \in \Sigma^*$, we~write $\mathsf{TIME}(\mathcal{M},w)$ for the \emph{maximal} length of~a~computation of~$\mathcal{M}$ on~$w$:
\begin{displaymath}
\mathsf{TIME}(\mathcal{M},w) := \mxm\{\lvert\gamma\rvert \mid \gamma \in C_w(\mathcal{M})\};
\end{displaymath} 
note that the maximum is well-defined, as all weighted Turing machines are assumed to be halting. Moreover, given $n \in \mathbb{N}$, we write
\begin{displaymath}
\mathsf{TIME}(\mathcal{M},n) := \mxm\{\mathsf{TIME}(\mathcal{M},w) \mid w \in \Sigma^*;~ \lvert w\rvert \leq n\}.
\end{displaymath}  

When $f\colon \mathbb{N} \to \mathbb{N}$ is a~function, we write $\mathsf{1NTIME}[S,\Sigma](f(n))$ for the set of all power series $r \in S\llangle\Sigma^*\rrangle$
such that $r = \|\mathcal{M}\|$ for~some weighted Turing machine $\mathcal{M}$ over~$S$~and~$\Sigma$ satisfying $\mathsf{TIME}(\mathcal{M},n) = O(f(n))$.
The~``$\mathsf{1}$''~stands for~single-tape machines in this notation. Moreover, let us define the~complexity class $\mathsf{1NTIME}[S](f(n))$ by
\begin{displaymath}
\mathsf{1NTIME}[S](f(n)) := \bigcup_{\text{$\Sigma$ is an alphabet}} \mathsf{1NTIME}[S,\Sigma](f(n)).
\end{displaymath}

The class $\mathbf{NP}[S]$, which is the counterpart of $\mathbf{NP}$ for formal power series over a~semiring~$S$, can now be defined in an expectable way;  
C.~Damm, M.~Holzer, and~P.~McKenzie~\cite{damm2002a} denote this class by $S\hyp\#\mathbf{P}$.

\begin{definition}
Let $S$ be a semiring. The class $\mathbf{NP}[S]$ is given by
\begin{displaymath}
\mathbf{NP}[S] := \bigcup_{k \in \mathbb{N}} \mathsf{1NTIME}[S](n^k).
\end{displaymath}
\end{definition}

It turns out that weighted Turing machines and the classes $\mathbf{NP}[S]$ can be used to capture a~large variety of~well-known complexity-theoretic settings, providing a~natural framework for their consistent study.
The~first five of the following examples were already observed by C.~Damm, M.~Holzer, and~P.~McKenzie~\cite{damm2002a}.        

\begin{example}
Weighted Turing machines over the Boolean semiring $\mathbb{B} = (\mathbb{B},\lor,\land,0,1)$ can obviously be identified with ordinary nondeterministic Turing machines without weights.
Hence, $\mathbf{NP}[\mathbb{B}]$ can be identified with the~usual complexity class $\mathbf{NP}$.
\end{example}

\begin{example}
\label{ex:sharpP}
Let $\mathcal{M} = (Q,\Gamma,\Delta,\sigma,q_0,F,\square)$ be a weighted Turing machine over the~semiring of~natural numbers $\mathbb{N} = (\mathbb{N},+,\cdot,0,1)$ and over some alphabet $\Sigma$.
When each transition is weighted by~$1$ -- \emph{i.e.}, $\sigma(e) = 1$ for~all $e \in \Delta$ -- the~coefficient of~each $w \in \Sigma^*$ in~$\|\mathcal{M}\|$  
is the number of~accepting computations of~$\mathcal{M}$~on~$w$. Moreover, each weighted Turing machine $\mathcal{M}$ over $\mathbb{N}$ is equivalent to~some other weighted Turing machine~$\mathcal{M}'$ 
such that all transitions of~$\mathcal{M}'$ are weighted by~$1$ and~$\mathcal{M}'$ runs in~polynomial time if~and~only~if~$\mathcal{M}$~does.
To~construct $\mathcal{M}'$ from $\mathcal{M}$, it is enough to replace each transition $e = (p,c,q,d,s)$ of $\mathcal{M}$ by transitions $\left(p,c,[e,1],c,0\right),\ldots,\left(p,c,[e,\sigma(e)],c,0\right)$ 
and~$\left([e,1],c,q,d,s\right),\ldots,\left([e,\sigma(e)],c,q,d,s\right)$, where $[e,1],\ldots,[e,\sigma(e)]$ are new states.
It follows that $\mathbf{NP}[\mathbb{N}]$ can be identified with the counting class $\#\mathbf{P}$ \cite{valiant1979b}. 
\end{example}\goodbreak

\begin{example}
A~weighted Turing machine $\mathcal{M}$ over the~finite field $\mathbb{F}_2 = \mathbb{Z}/2\mathbb{Z}$ and~over an~alphabet~$\Sigma$
can~only contain transitions weighted by~$1$. It is then immediate that $\|\mathcal{M}\|$ is a~power series in~$\mathbb{F}_2\llangle\Sigma^*\rrangle$ such that
$(\|\mathcal{M}\|,w)$ is, for~each $w \in \Sigma^*$, the~parity of~the~number of accepting computations of~$\mathcal{M}$ on~$w$. As~a~result, $\mathbf{NP}[\mathbb{Z}_2]$ can be identified 
with the complexity class $\oplus\mathbf{P}$ -- \emph{i.e.}, ``parity-$\mathbf{P}$'' \cite{papadimitriou1983a}.
\end{example}

\begin{example}
Let $k \geq 2$ be a~natural number. Then similarly as above, a~weighted Turing machine $\mathcal{M}$ over the~ring $\mathbb{Z}/k\mathbb{Z}$ with all transitions weighted by~$1$
realises a~power series $\|\mathcal{M}\|$ such that $(\|\mathcal{M}\|,w)$ is, for~each~$w \in \Sigma^*$, the~congruence class of~the~number of~accepting computations of~$\mathcal{M}$ on~$w$ modulo~$k$.
Moreover, the~same reasoning as in~Example~\ref{ex:sharpP} can be used to observe that every weighted Turing machine over $\mathbb{Z}/k\mathbb{Z}$ is equivalent to some machine with all transitions
weighted by $1$. It follows that $\mathbf{NP}[\mathbb{Z}/k\mathbb{Z}]$ is~precisely the~class of~all series over $\mathbb{Z}/k\mathbb{Z}$ with support in~$\mathbf{MOD_{k}P}$~\cite{beigel1992a,cai1990a}. 
\end{example}

\begin{example}
The class $\mathbf{GapP}$, introduced in \cite{gupta1991a,gupta1995a} as $\mathbb{Z}\#\mathbf{P}$, and independently in \cite{fenner1994a}, is
the~closure of~the~class $\#\mathbf{P}$ under subtraction. It can be equivalently described as the class consisting of all functions $\mathrm{gap}_{\mathcal{M}}\colon\Sigma^* \to \mathbb{Z}$, for some alphabet $\Sigma$,
such that $\mathrm{gap}_{\mathcal{M}}(w)$ is the difference between the number of~accepting and~rejecting computations of~some polynomial-time nondeterministic Turing machine~$\mathcal{M}$ on~$w \in \Sigma^*$~\cite{fenner1994a}.
Given any polynomial-time nondeterministic Turing machine~$\mathcal{M}$, one can construct a~polynomial-time weighted Turing machine $\mathcal{M}'$ over the~ring of~integers~$\mathbb{Z}$
that first simulates $\mathcal{M}$ using transitions weighted by~$1$. If $\mathcal{M}$ accepts, $\mathcal{M}'$ makes a single step using a transition weighted by $1$ and accepts;
if $\mathcal{M}$ rejects, $\mathcal{M}'$ makes a single step using a transition weighted by~$-1$ and \emph{accepts} as well. Clearly $(\|\mathcal{M}'\|,w) = \mathrm{gap}_{\mathcal{M}}(w)$ for~all~$w \in \Sigma^*$.
Conversely, a~reasoning similar to the one of~Example~\ref{ex:sharpP} shows that every weighted Turing machine $\mathcal{M}$ over~$\mathbb{Z}$ can be assumed to 
only contain transitions weighted by $1$ or $-1$. Given~$\mathcal{M}$ like this, one can construct a~nondeterministic Turing machine $\mathcal{M}'$ that simulates
$\mathcal{M}$ and~maintains in~state the~product of~weights of~the~transitions used in~the~computation so far -- which is always either~$1$, or~$-1$.
If~the~computation of~$\mathcal{M}$ accepts with value~$1$, then $\mathcal{M}'$ accepts; if it accepts with value~$-1$, $\mathcal{M}'$ rejects; if it rejects,
$\mathcal{M}'$ nondeterministically branches into two states such that one of~them is accepting and~the~other one is rejecting. Clearly $\mathrm{gap}_{\mathcal{M}'}(w) = (\|\mathcal{M}\|,w)$ 
for~all $w \in \Sigma^*$, so that $\mathbf{NP}[\mathbb{Z}]$ can be identified with $\mathbf{GapP}$. 
\end{example}

We now identify several new examples of~known complexity classes that can be described as~$\mathbf{NP}[S]$ for~a~suitable semiring~$S$.

\begin{example}
\emph{Fuzzy Turing machines} in the sense of \cite{wiedermann2004a} can be viewed as weighted Turing machines over a semiring as well.
A~\emph{triangular norm} (or \emph{t-norm}, for short) is an associative and commutative binary operation $\ast$ on~the~real interval $[0,1]$ that is nondecreasing -- \emph{i.e.}, $x_1 \ast y \leq x_2 \ast y$ whenever $x_1,x_2,y \in [0,1]$ are such that $x_1 \leq x_2$ --
and the~equality $1 \ast x = x$ is satisfied for all $x$ from the interval in~consideration~\cite{klement2000a}. It~follows that $\ast$ distributes over $\mxm$ and $0 \ast x = 0$ holds for all $x \in [0,1]$. Thus, $F_{\ast} = ([0,1],\mxm,\ast,0,1)$ is a~semiring for each t-norm~$\ast$.
Every fuzzy Turing machine $\mathcal{M}$, as~understood in~\cite{wiedermann2004a}, can then be viewed as~a~weighted Turing machine over the~semiring $F_{\ast}$ for~some t-norm $\ast$.
Coefficients of~particular words in~$\|\mathcal{M}\|$ represent their degrees of~membership to the~fuzzy language realised by~$\mathcal{M}$.
The~class~$\mathbf{NP}[F_{\ast}]$ can~thus~be interpreted as~the~class of~all \emph{fuzzy} languages realisable by fuzzy Turing machines with t-norm $\ast$ in~polynomial time.
(Note that fuzzy machines can also be seen as acceptors for ``ordinary'' languages~\cite{wiedermann2004a}.)   
\end{example}

\begin{example}
\label{ex:transducers}
Let $\Sigma_1$ and $\Sigma_2$ be alphabets, and $\mathcal{M}$ a weighted Turing machine with input alphabet $\Sigma_1$ over the~semiring of~finite languages \smash{$2_{\mathit{fin}}^{\Sigma_2^*} = (2_{\mathit{fin}}^{\Sigma_2^*},\cup,\cdot,\emptyset,\{\eps\})$}.
Then $\|\mathcal{M}\|$ can be viewed as~a~multivalued function assigning a~finite subset of~$\Sigma_2^*$ to~each $w \in \Sigma_1^*$.
It is easy to see that $\mathcal{M}$ can be turned into a~nondeterministic transducer machine -- in the sense of~\cite{book1985a} -- such that there is 
an~accepting computation on~$u \in \Sigma_1^*$ with the~output tape containing a~word $v \in \Sigma_2^*$ at~its end if~and~only~if $v \in (\|\mathcal{M}\|,u)$.
Conversely, given a~nondeterministic transducer machine realising some multivalued function, it is possible to construct an~``equivalent'' weighted Turing machine 
over \smash{$2_{\mathit{fin}}^{\Sigma_2^*}$}. What needs to be done is to first simulate the transducer on~a~single tape of~a~weighted Turing machine with polynomial overhead, while 
the~transitions responsible for~this simulation are weighted by~$\{\eps\}$. After the~simulation finally reaches an~accepting state of~the~original transducer,
the~simulating weighted machine just goes over the output from left to right using a~transition weighted by $\{c\}$ whenever $c \in \Sigma_2$ is being read; only then the~simulating machine accepts.   
This in particular implies that there is a~weighted Turing machine realising a~series \smash{$r \in 2_{\mathit{fin}}^{\Sigma_2^*}\llangle\Sigma_1^*\rrangle$} in polynomial time if and only
if there is a~polynomial-time nondeterministic transducer machine realising $r$ interpreted as a multivalued function. As a~result, the union of~\smash{$\mathbf{NP}[2_{\mathit{fin}}^{\Sigma^*}]$} over all alphabets $\Sigma$ can be interpreted
as the class $\mathbf{NPMV}$ of~all multivalued functions realised by nondeterministic polynomial-time transducer machines~\cite{book1985a}.\goodbreak 
\end{example}

\begin{example}
Multivalued functions of the previous example are, in principle, set-valued functions. 
In~case the~free semiring $\mathbb{N}\langle\Sigma_2^*\rangle = (\mathbb{N}\langle\Sigma_2^*\rangle,+,\cdot,0,1)$ is used instead of \smash{$2_{\mathit{fin}}^{\Sigma_2^*}$}, then 
the setting is changed from set-valued functions to ``multiset-valued'' functions computed by ``nondeterministic transducer machines with counting'',
\emph{i.e.}, transducer machines in which each accepting computation with input $u \in \Sigma_1^*$ and~output $v \in \Sigma_2^*$ adds $1$ to the multiplicity of~$v$ in~the~finite multiset corresponding to~$u$.
The~union of~$\mathbf{NP}[\mathbb{N}\langle\Sigma^*\rangle]$ over all alphabets~$\Sigma$ then corresponds to the class of all multiset-valued functions computed
by~``nondeterministic polynomial-time transducer machines with counting''. As we observe later in this article, weighted computation over free semirings 
can be described as ``hardest'' among all semirings (cf.~Proposition~\ref{prop:division}).   
\end{example}

\begin{example}
The \emph{radix order} $\preceq$ on $\{0,1\}^*$ is defined for all $x,y \in \{0,1\}^*$ by $x \preceq y$ if and only if either $\lvert x\rvert < \lvert y\rvert$, or $\lvert x\rvert = \lvert y\rvert$ and $x$ is smaller than or equal to $y$ according to the lexicographic order.
Let us consider the semiring $S_{\mxm} = (\{0,1\}^* \cup \{-\infty\},\mxm,\cdot,-\infty,\eps)$ such that the~restriction of~$\mxm$~to~$\{0,1\}^*$ is the~maximum according 
to the~radix order~$\preceq$,
$\mxm(x,-\infty) = \mxm(-\infty,x) = x$ for all $x \in S_{\mxm}$, the~restriction of~$\cdot$~to~$\{0,1\}^*$ is the usual concatenation operation,
and~$x \cdot (-\infty) = (-\infty) \cdot x = -\infty$ for~all~$x \in S_{\mxm}$. If $\mathrm{num}(x)$ denotes a~number with binary representation $x \in \{0,1\}^*$, then 
$\varphi\colon \{0,1\}^* \to \mathbb{N}$, defined by~$\varphi\colon x \mapsto \mathrm{num}(1x) - 1$, is clearly an~isomorphism of linearly ordered sets $(\{0,1\}^*,\preceq)$ and~$(\mathbb{N},\leq)$.

Given a~weighted Turing machine $\mathcal{M}$ over $S_{\mxm}$ with input alphabet $\Sigma$, it is straightforward to construct
a~nondeterministic Turing machine $\mathcal{M}'$ with binary encoded nonnegative integer output such that
an~accepting configuration with output $\varphi(x)$ for $x \in \{0,1\}^*$ is reachable in $\mathcal{M}'$ upon an~input~word~$w$
if~and~only~if there is an accepting computation $\gamma$ of the weighted Turing machine $\mathcal{M}$ such that $\lambda(\gamma) = w$ and $\sigma(\gamma) = x$. 
     
A converse construction can be done similarly as in Example \ref{ex:transducers}.
Moreover, it is easy to see that both constructions result in at most polynomial overhead.   

Hence, as the maximum is used as an~additive operation of~the~semiring~$S_{\mxm}$, the~class $\mathbf{NP}[S_{\mxm}]$ corresponds \emph{in principle}
to maximisation problems in the class $\mathbf{OptP}$ of \cite{krentel1988a}:
the class $\mathbf{OptP}$ consists of~all problems, in~which the objective is to compute the~\emph{value} of~an~optimal solution to
a~given instance of~an~optimisation problem in $\mathbf{NPO}$ \cite{kann1992a,kann1995a,crescenzi1999a};
a~minor difference between this class and~the~$S_{\mxm}$-weighted setting is that the value $-\infty$ is not permitted for problems in $\mathbf{OptP}$.     

In addition, one can also describe the~class of~minimisation problems in~$\mathbf{OptP}$ using the semiring $S_{\mnm} = (\{0,1\}^* \cup \{\infty\},\mnm,\cdot,\infty,\eps)$, for which the operations are defined in a similar way as for $S_{\mxm}$.  
\end{example}

\begin{example}
Considering the \emph{max-plus} -- or \emph{arctic}~\cite{droste2009b} -- \emph{semiring} $\mathbb{N}_{\mxm} = (\mathbb{N} \cup \{-\infty\},\mxm,+,-\infty,0)$ of~natural numbers instead of~$S_{\mxm}$
has an~effect of~requiring unary outputs in~the~corresponding nondeterministic Turing machine instead of~binary outputs. The~complexity class $\mathbf{NP}[\mathbb{N}_{\mxm}]$
thus roughly corresponds to~maximisation problems from the class $\mathbf{OptP}[O(\log n)]$ of~\cite{krentel1988a}, which is the same to $\mathbf{NPO~PB}$ \cite{kann1992a,kann1995a,crescenzi1999a}
as $\mathbf{OptP}$ is to $\mathbf{NPO}$. 

Moreover, the class of~minimisation problems from $\mathbf{OptP}[O(\log n)]$ can be captured via the \emph{tropical semiring} $\mathbb{N}_{\mnm} = (\mathbb{N} \cup \{\infty\},\mnm,+,\infty,0)$.  
\end{example}

In~addition to $\mathbf{NP}[S]$, other time complexity classes of power series can be introduced as well by generalising the~usual definition of~a~nondeterministic time complexity class to weighted Turing machines --
for~instance,
\begin{displaymath}
\mathbf{NEXPTIME}[S] := \bigcup_{k \in \mathbb{N}} \mathsf{1NTIME}[S]\left(2^{n^k}\right).
\end{displaymath}
Similarly, weighted space complexity classes can be defined using weighted multitape Turing machines with a~read-only input tape.\goodbreak

\section{$\mathbf{FP}[S]$: Complexity Classes of Semiring-Valued Functions}

By analogy to the relation between $\mathbf{P}$ and $\mathbf{NP}$, the~class $\mathbf{FP}$ \cite{arora2009a} is usually considered to be the proper ``deterministic counterpart'' of $\#\mathbf{P}$ in the presence of counting.
In~this context, $\mathbf{FP}$ can be described as the~class of~all functions~$f$ computable by polynomial-time deterministic Turing machines that, given
an~input word~$w$, return some natural number $f(w)$ encoded in~binary; it~is known that $\mathbf{FP} \subseteq \#\mathbf{P}$ and~that the~equality $\mathbf{FP} = \#\mathbf{P}$ would imply $\mathbf{P} = \mathbf{NP}$.

In a similar way, we now define a ``deterministic counterpart'' $\mathbf{FP}[S]$ to the class $\mathbf{NP}[S]$ for an~arbitrary semiring~$S$; as we observe, this coincides with $\mathbf{P}$
in case $S = \mathbb{B}$ and with $\mathbf{FP}$ in case $S = \mathbb{N}$. As~a~basic ingredient for the definition of $\mathbf{FP}[S]$, we use the~complexity classes
$\mathbf{FP}_G[S,\Sigma]$ for a~finite alphabet $\Sigma$ and~a~finite subset~$G$ of~$S$. The class $\mathbf{FP}_G[S,\Sigma]$ consists of all functions $f$ from $\Sigma^*$ to
the subsemiring $\langle G\rangle$ of $S$ generated by $G$, for which there is a~polynomial-time deterministic Turing machine that, given $w \in \Sigma^*$, outputs some word in~the~algebra of~terms~$T(G)$
that evaluates to~$f(w)$ in the semiring~$S$. The~class $\mathbf{FP}[S]$ is then defined to be the union of $\mathbf{FP}_G[S,\Sigma]$ over all such $G$ and $\Sigma$.
As~functions from~$\Sigma^*$ to~$S$ are formally the~same objects as power series over~$\Sigma$ with coefficients in~$S$, it is reasonable to ask about the~relationship of~the~class $\mathbf{FP}[S]$ to $\mathbf{NP}[S]$ --
and~indeed, we easily obtain the~inclusion $\mathbf{FP}[S] \subseteq \mathbf{NP}[S]$ for~all semirings~$S$.              

\begin{definition}
Let $S$ be a semiring, $G$ a finite subset of $S$, and $\Sigma$ an alphabet. The~class $\mathbf{FP}_G[S,\Sigma]$ consists of all functions $f\colon \Sigma^* \to \langle G\rangle$,
for which there exists a polynomial-time multitape deterministic Turing machine with output that transforms each input word $w \in \Sigma^*$ to some term $t(w) \in T(G)$ satisfying $h_G[S](t(w)) = f(w)$.
\end{definition}

The~following lemma shows that the class $\mathbf{FP}_G[S,\Sigma]$ is actually the same for all $G$ generating the~same subsemiring of~$S$.

\begin{lemma}
\label{le:generators}
Let $S$ be a semiring, $G,H$ finite subsets of $S$, and $\Sigma$ an alphabet. If $\langle G\rangle \subseteq \langle H\rangle$, then $\mathbf{FP}_G[S,\Sigma] \subseteq \mathbf{FP}_H[S,\Sigma]$. As a result, $\mathbf{FP}_G[S,\Sigma] = \mathbf{FP}_H[S,\Sigma]$ whenever $\langle G\rangle = \langle H\rangle$. 
\end{lemma}
\begin{proof}
Let $f\colon \Sigma^* \to \langle G\rangle$ be in $\mathbf{FP}_G[S,\Sigma]$. Let $\mathcal{M}$ be a polynomial-time deterministic Turing machine that outputs a~term $t(w) \in T(G)$ such that $h_G[S](t(w)) = f(w)$ for~each input word $w \in \Sigma^*$.
As $\langle G\rangle \subseteq \langle H\rangle$, there is a~term $t_g \in T(H)$ for each $g \in G$ such that $h_H[S](t_g) = g$.
Let $\mathcal{M}'$ be a~deterministic Turing machine that first simulates $\mathcal{M}$ on each $w \in \Sigma^*$, so that it outputs $t(w)$ on~a~working tape.
After this is done, $\mathcal{M}'$ copies $t(w)$ to the output tape, while replacing each $g \in G$ by the term $t_g$ (which is of~length independent of~$w$). 
It is clear that the resulting term $t'(w)$ satisfies $h_H[S](t'(w)) = h_G[S](t(w)) = f(w)$ and that~$\mathcal{M}'$ operates in polynomial time. Hence, $f$ is in $\mathbf{FP}_H[S,\Sigma]$.
\end{proof}

\begin{definition}
Let $S$ be a semiring. For each alphabet $\Sigma$, let
\begin{displaymath}
\mathbf{FP}[S,\Sigma] := \bigcup_{\substack{G \subseteq S \\ \text{$G$ finite}}} \mathbf{FP}_G[S,\Sigma].  
\end{displaymath}
Similarly, given a finite subset $G$ of $S$, let
\begin{displaymath}
\mathbf{FP}_G[S] := \bigcup_{\text{$\Sigma$ is an alphabet}} \mathbf{FP}_G[S,\Sigma].
\end{displaymath}
We then write
\begin{displaymath}
\mathbf{FP}[S] := \bigcup_{\text{$\Sigma$ is an alphabet}} \mathbf{FP}[S,\Sigma] = \bigcup_{\substack{G \subseteq S \\ \text{$G$ finite}}} \mathbf{FP}_G[S].
\end{displaymath}
\end{definition}

It follows by Lemma \ref{le:generators} that $\mathbf{FP}[S] = \mathbf{FP}_G[S]$ when $S$ is a semiring is finitely generated by $G$.\goodbreak

\begin{example}
Let $S = \mathbb{B}$. As~evaluation of Boolean expressions can be done in polynomial time, each $f \in \mathbf{FP}[\mathbb{B}]$ corresponds to a problem in $\mathbf{P}$.
Conversely, each problem in $\mathbf{P}$ can be turned into a function in~$\mathbf{FP}[\mathbb{B}]$ by outputting the term $1$ if an input is accepted and the term $0$ if it is rejected.
The class $\mathbf{FP}[\mathbb{B}]$ can thus be identified with $\mathbf{P}$. 
\end{example}

\begin{example}
The situation is similar for $S = \mathbb{N}$. Here $\mathbf{FP}[\mathbb{N}] = \mathbf{FP}_{\emptyset}[\mathbb{N}]$, as $\langle\emptyset\rangle = \mathbb{N}$. Outputs of~deterministic Turing machines
corresponding to functions in this class take the~form of~arithmetic expressions consisting of~the~constants~$0$ and~$1$ and~operators $+$ and $\cdot$.
Such expressions can be evaluated in polynomial time in binary -- every function in~$\mathbf{FP}[\mathbb{N}]$ thus also is in~$\mathbf{FP}$.
On the other hand, when a~function $f\colon \Sigma^* \to \mathbb{N}$ is in $\mathbf{FP}$, there exists a~deterministic Turing machine that outputs, for each $w \in \Sigma^*$,
a binary representation $\mathrm{bin}(f(w))$ of $f(w)$. However, $\mathrm{bin}(f(w))$ can be transformed in~polynomial time into a~term $t(w) \in T(\emptyset)$ such that $h_{\emptyset}[\mathbb{N}](t(w)) = f(w)$ --
it suffices to observe that there is a~polynomial function $p\colon \mathbb{N} \to \mathbb{N}$ independent of~$w$ and a~set $K(w) \subseteq \{0,\ldots,p(\lvert w\rvert)\}$ such that
the~length of~$\mathrm{bin}(f(w))$ is at~most $p(\lvert w\rvert)+1$ and~the~$k$-th order bit of~$\mathrm{bin}(f(w))$ is~$1$ if~and~only~if $k \in K(w)$.
Then 
\begin{displaymath}
f(w) = \sum_{k \in K(w)} 2^k = \sum_{k \in K(w)} (1 + 1)^k = \sum_{k \in K(w)} \prod_{i = 1}^k (1 + 1).
\end{displaymath}      
The~term $t(w)$ can thus be constructed from $\mathrm{bin}(f(w))$ in~polynomial time as~a~sum of~at~most $p(\lvert w\rvert) + 1$ products of~at~most $p(\lvert w\rvert)$ terms $(1 + 1)$.
Hence $f$ is in~$\mathbf{FP}[\mathbb{N}]$ and~$\mathbf{FP}$ coincides with~$\mathbf{FP}[\mathbb{N}]$.  
\end{example}

\begin{proposition}
\label{prop:fpinclnp}
Let $S$ be a semiring. Then $\mathbf{FP}[S] \subseteq \mathbf{NP}[S]$.
\end{proposition}
\begin{proof}
In~case an~$\mathbf{FP}[S]$-machine $\mathcal{M}$ exists for~$r \in S\llangle\Sigma^*\rrangle$, an equivalent $\mathbf{NP}[S]$-machine~$\mathcal{M}'$ can first simulate~$\mathcal{M}$ on~its~input word~$w$, using transitions weighted by $1$ (and with at most polynomial overhead). 
Next, it can use the~output of~the~machine~$\mathcal{M}$ -- \emph{i.e.}, a~term $t \in T(G)$ for~some finite $G \subseteq S$ -- to assure that~$(\|\mathcal{M}'\|,w) = h_G[S](t)$.
It~is trivial to do so when $t$ is in~$G \cup \{0,1\}$. In~case $t = (t_1 + t_2)$ for some $t_1,t_2 \in T(G)$, then $\mathcal{M}'$ nondeterministically ``chooses'' one of~these terms and~evaluates it recursively;
for~$t = (t_1 \cdot t_2)$, the~machine~$\mathcal{M}'$ evaluates $t_1$ followed by $t_2$. Correctness follows by~distributivity of~$S$.       
\end{proof}

\section{Reductions and $\mathbf{NP}[S]$-Complete Problems}

Let us now consider reductions between problems represented by formal power series and completeness for the classes $\mathbf{NP}[S]$.
These concepts were actually already studied in~the~framework of~algebraic \mbox{Turing} machines and~the~resulting complexity classes~\cite{beaudry2007a,damm2002a}.
In~particular, M.~Beaudry and~M.~Holzer~\cite{beaudry2007a} managed to prove $\mathbf{NP}[S]$-completeness of~the~evaluation problem for~scalar tensor formulae
over additively finitely generated semirings~$S$ under polynomial-time many-one reductions. In~what follows, we show that two fundamental $\mathbf{NP}$-completeness 
results do actually generalise to the~weighted setting: for~every finitely generated semi\-ring~$S$, we prove $\mathbf{NP}[S]$-completeness
of~a~problem $\mathsf{WTMSAT}[S]$ inspired by $\mathsf{TMSAT}$ of~\cite{arora2009a}, as~well~as of~$\mathsf{SAT}[S]$, a~generalisation of~$\mathsf{SAT}$ to~weighted propositional logics;
the~latter observation generalises the~Cook-Levin theorem. 

Several different notions of reduction seem to be reasonable in~the~setting of~the~weighted classes $\mathbf{NP}[S]$ -- in fact, this is already so in~the~particular case of~counting problems, which
can be identified with series over the~semiring of~natural numbers. We mostly consider the~weakest among these reductions resulting in~the~strongest completeness requirement -- over~$\mathbb{N}$,
this is precisely the~\emph{parsimonious} reduction between counting problems. Over a~general semiring, we use the term \emph{polynomial-time many-one reduction} instead.

\begin{definition}
Let $S$ be a semiring and $\Sigma_1,\Sigma_2$ alphabets. A problem -- \emph{i.e.}, a~power series -- $r \in S\llangle\Sigma_1^*\rrangle$ is \emph{polynomially many-one reducible} to a~series $s \in S\llangle\Sigma_2^*\rrangle$,
written $r \leq_m s$, if there is a~function $f\colon \Sigma_1^* \to \Sigma_2^*$ computable deterministically in~polynomial time such that $(s,f(w)) = (r,w)$ for~all $w \in \Sigma_1^*$.
\end{definition}

As usual, we say that $s \in S\llangle\Sigma^*\rrangle$ is \emph{$\mathbf{NP}[S]$-hard} with respect to polynomial-time many-one reductions if~$r \leq_m s$ for~all $r \in \mathbf{NP}[S]$;
a~series $s$ is \emph{$\mathbf{NP}[S]$-complete} with respect to $\leq_m$ if it belongs to $\mathbf{NP}[S]$ and~at~the~same time, it is~$\mathbf{NP}[S]$-hard under $\leq_m$.\goodbreak 

\begin{theorem}
\label{th:complete}
Let $S$ be a semiring. An~$\mathbf{NP}[S]$-complete problem with respect to~$\leq_m$ exists if and only if $S$ is finitely generated. 
\end{theorem}
\begin{proof}
Let us first assume that $S$ is finitely generated, and let $G \subseteq S$ be a finite set such that $\langle G\rangle = S$. Each $a \in S$ can then be encoded by a finite word 
\begin{displaymath}
\tau_G(a) = g_{1,1} \cdot \ldots \cdot g_{1,m_1} + \ldots + g_{k,1} \cdot \ldots \cdot g_{k,m_k}
\end{displaymath}
with $k,m_1,\ldots,m_k \in \mathbb{N} \setminus \{0\}$ and $g_{i,j} \in G \cup \{0,1\}$ for $i = 1,\ldots,k$ and $j = 1,\ldots,m_i$, such that
\begin{displaymath}
g_{1,1} \ldots g_{1,m_1} + \ldots + g_{k,1} \ldots g_{k,m_k} = a
\end{displaymath}
holds in $S$. Every weighted Turing machine $\mathcal{M}$ over $S$ and $\Sigma = \{0,1\}$ thus admits an~effective encoding~$\langle\mathcal{M}\rangle$ -- the weight~$\sigma(e)$ of each transition $e$ is encoded via $\tau_G(\sigma(e))$
and~the~rest can be done in~the~same way as~for~nondeterministic machines without weights. It is straightforward to construct a~universal weighted Turing machine $\mathcal{U}$ over~$S$ such that
$\left(\|\mathcal{U}\|,\langle\mathcal{M}\rangle\#w\right) = (\|\mathcal{M}\|,w)$ holds for~all weighted machines $\mathcal{M}$ over~$S$ with input alphabet $\Sigma = \{0,1\}$ and~all~$w \in \Sigma^*$,
while 
\begin{displaymath}
\mathsf{TIME}\left(\mathcal{U},\langle\mathcal{M}\rangle\#w\right) \leq p\left(\lvert\langle\mathcal{M}\rangle\rvert + \lvert w\rvert\right) \cdot \mathsf{TIME}\left(\mathcal{M},w\right)^{k} + q\left(\lvert\langle\mathcal{M}\rangle\rvert + \lvert w\rvert\right)
\end{displaymath}
for some constant $k \in \mathbb{N}$ and~polynomial functions $p,q\colon \mathbb{N} \to \mathbb{N}$ independent of~$\mathcal{M}$~and~$w$.    

Next, let us consider the problem -- a~power series -- $\mathsf{WTMSAT}[S]$ inspired by~$\mathsf{TMSAT}$ for~Turing machines without weights \cite{arora2009a} and~given as follows: 
upon input of~the~form $\langle\mathcal{M}\rangle\#w\#1^m$ for~some weighted Turing machine $\mathcal{M}$ over~$S$ with input alphabet $\Sigma = \{0,1\}$, some $w \in \Sigma^*$, and some $m \in \mathbb{N}$,
let 
\begin{displaymath}
\left(\mathsf{WTMSAT}[S],\langle\mathcal{M}\rangle\#w\#1^m\right) := \sum_{\substack{\gamma \in A(\mathcal{M}) \\ \lambda(\gamma) = w \\ \lvert \gamma\rvert \leq m}} \sigma(\gamma).
\end{displaymath}
This in particular implies that $\left(\mathsf{WTMSAT}[S],\langle\mathcal{M}\rangle\#w\#1^m\right) = (\|\mathcal{M}\|,w)$ whenever $\mathsf{TIME}(\mathcal{M},w) \leq m$ holds.
Moreover, set $(\mathsf{WTMSAT}[S],x) := 0$ for all other words $x$ over the~input alphabet of~$\mathsf{WTMSAT}[S]$.

It is not hard to see that $\mathsf{WTMSAT}[S]$ is in $\mathbf{NP}[S]$ -- one can simulate $m$ steps of the~machine~$\mathcal{M}$ using the~universal machine~$\mathcal{U}$ upon input $\langle\mathcal{M}\rangle\#w$.
On~the~other hand, if a~series $r \in S\llangle\Sigma^*\rrangle$ for~$\Sigma = \{0,1\}$ belongs to~$\mathbf{NP}[S]$ -- so that there is a weighted Turing machine $\mathcal{M}$ over $S$ and~$\Sigma$ such
that $\|\mathcal{M}\| = r$ and~$\mathsf{TIME}(\mathcal{M},n) \leq cn^k + d =: p(n)$ for~some constants $c,k,d \in \mathbb{N}$ and~all~$n \in \mathbb{N}$ -- then
\begin{displaymath}
\left(\mathsf{WTMSAT}[S],\langle\mathcal{M}\rangle\#w\#1^{p(\lvert w\rvert)}\right) = (r,w)
\end{displaymath}
for all $w \in \Sigma^*$. Hence, clearly $r \leq_m \mathsf{WTMSAT}[S]$.
Moreover, series in $\mathbf{NP}[S]$ over other alphabets can be reduced to a~series over $\Sigma = \{0,1\}$ via encoding and~$\leq_m$ is clearly
transitive. This proves that $\mathsf{WTMSAT}[S]$ is $\mathbf{NP}[S]$-complete.

It remains to prove that $\mathbf{NP}[S]$ has no complete problems with respect to $\leq_m$ when $S$ is not finitely generated. Suppose for~the~purpose of~contradiction that
$S$ is not finitely generated and~that $r \in S\llangle\Sigma^*\rrangle$, for~some alphabet $\Sigma$, is $\mathbf{NP}[S]$-complete. Then there is a~weighted Turing machine $\mathcal{M} = (Q,\Gamma,\Delta,\sigma,q_0,F,\square)$
over~$S$ and~$\Sigma$ such that $\|\mathcal{M}\| = r$ and~it is easy to see that the~coefficients of~$r$ in~fact belong to the~finitely generated subsemiring $\langle G\rangle$ of~$S$ for 
$G = \{\sigma(e) \mid e \in \Delta\}$. Hence, there is an~element $a \in S \setminus \langle G\rangle$ and~it~follows that, e.g., the~series $\sum_{w \in \Sigma^*} aw$ from $\mathbf{NP}[S]$ cannot be many-one reduced to $r$.         
\end{proof}

We now proceed to describe, for~each finitely generated semiring~$S$, a~slightly more interesting example of~an~$\mathbf{NP}[S]$-complete problem: a~weighted generalisation of~the~Boolean satisfiability problem~$\mathsf{SAT}$, which we call~$\mathsf{SAT}[S]$. 
This is the~``satisfiability'' problem for~\emph{weighted propositional logics over semirings} -- that~is, for~propositional fragments of~the~weighted MSO~logics of~M.~Droste and~P.~Gastin~\cite{droste2007a,droste2009c}.
Basically, the~weighted propositional logic over~$S$ is obtained from the~usual propositional logic by~incorporating constants from~$S$ and~defining the semantics 
of~logical connectives $\lor$ and~$\land$ via the~operations~$+$ and~$\cdot$ of~the~underlying semiring~$S$; negation is permitted for~propositional variables only, for which it has the~usual semantics.
Each~weighted propositional formula $\varphi$ thus admits a~value in~$S$ as~its semantics for a~fixed assignment of~truth values to~propositional variables.\goodbreak 

Let us fix an infinite alphabet $X$ of all propositional variables for the rest of this section; all variables occurring in~our constructions below are assumed to be taken from~$X$.
The~language of~the~\emph{weighted propositional logic} over a~semiring~$S$ is built upon an infinite alphabet consisting of~all symbols in~$X$,
all elements of $S$, symbols ``$\lor$'', ``$\land$'', and ``$\lnot$'' for~logical connectives, and~symbols ``$($'' and~``$)$'' for~parentheses.    
The~\emph{syntax} of~the~weighted propositional logic over~$S$ is defined as~follows:
\begin{enumerate}
\item{For each propositional variable $x \in X$, both $x$ and $\lnot x$ are well-formed propositional formulae over $S$. Moreover, each $a \in S$ is a well-formed propositional formula over $S$. 
}
\item{Let $\varphi,\psi$ be well-formed propositional formulae over $S$. Then $(\varphi \lor \psi)$ and $(\varphi \land \psi)$ are well-formed propositional formulae over $S$ as well.
}
\item{Nothing else is a well-formed propositional formula over $S$.
}
\end{enumerate}
When the~underlying semiring~$S$ is finitely generated, we denote by~$\langle \varphi\rangle$ an~effective encoding of~a~formula~$\varphi$ into some finite alphabet $\Sigma$ independent of~$\varphi$; in~particular, variables are represented by~binary numbers and~elements of~$S$ are encoded
in~terms of~generators of~$S$ as~in~the~proof of~Theorem~\ref{th:complete}.

A \emph{truth assignment} is a mapping $V\colon X \to \{0,1\}$. The \emph{semantics} of propositional formulae over $S$ with respect to the truth assignment $V$ is defined as follows:   
\begin{enumerate}
\item{For each propositional variable $x \in X$, let $\overline{V}(x) = V(x)$; moreover, let $\overline{V}(\lnot x) = 1$ if $V(x) = 0$ and~$\overline{V}(\lnot x) = 0$ if $V(x) = 1$.
In~addition, let $\overline{V}(a) = a$ for~each $a \in S$. 
}
\item{For each two well-formed propositional formulae $\varphi,\psi$ over the~semiring~$S$, let $\overline{V}(\varphi \lor \psi) = \overline{V}(\varphi) + \overline{V}(\psi)$ 
and~$\overline{V}(\varphi \land \psi) = \overline{V}(\varphi) \cdot \overline{V}(\psi)$.
}
\end{enumerate}
Values $\overline{V}(\varphi)$ depend just on the restriction of $V$ to the set $X_{\varphi}$ of all variables $x \in X$ occurring in~$\varphi$.
In~case this restriction is given by a~mapping $W\colon X_{\varphi} \to \{0,1\}$, we also write $\overline{W}(\varphi)$ for~$\overline{V}(\varphi)$.

Observe that an equivalent of the usual propositional logic is obtained when the~Boolean semiring~$\mathbb{B}$ is~taken for~$S$.
\smallskip 

The problem $\mathsf{SAT}[S]$ is to determine, for a given weighted propositional formula~$\varphi$ over~a~finitely generated semiring~$S$ represented by its encoding $\langle\varphi\rangle$, the sum of~values of~$\varphi$ over all choices of~truth values of~variables from~$X_{\varphi}$.
The series $\mathsf{SAT}[S] \in S\llangle\Sigma^*\rrangle$ is thus given by 
\begin{displaymath}
(\mathsf{SAT}[S],\langle\varphi\rangle) = \sum_{W \in \{0,1\}^{X_{\varphi}}} \overline{W}(\varphi)
\end{displaymath}
for all formulae~$\varphi$ and by $(\mathsf{SAT}[S],w) = 0$ for all $w \in \Sigma^*$ that do~not encode weighted propositional formulae over~$S$. We now prove that this problem is $\mathbf{NP}[S]$-complete with respect to~polynomial-time many-one reductions. 

\begin{theorem}
\label{th:sat}
Let $S$ be a finitely generated semiring. Then $\mathsf{SAT}[S]$ is $\mathbf{NP}[S]$-complete with respect to $\leq_m$. 
\end{theorem}
\begin{proof}
It is not hard to see that $\mathsf{SAT}[S]$ is in $\mathbf{NP}[S]$. Upon an~input word $\langle\varphi\rangle$, the~weighted Turing machine for~$\mathsf{SAT}[S]$
first ``guesses'' the assignment $W \in \{0,1\}^{X_{\varphi}}$ and then evaluates~$\varphi$ according to $\overline{W}$, in~the~sense that the~sum of~values of~``computation suffixes'' started in~the~given configuration is $\overline{W}(\varphi)$. 
This evaluation is done similarly as in the proof of Proposition \ref{prop:fpinclnp}: 
it is trivial for~formulae $x$ and $\lnot x$ with $x \in X_{\varphi}$ -- the~machine accepts the~formula if and only if its truth value is~$1$, while this is done using transitions weighted by~$1$;
the~evaluation of~$a \in S$ can easily be done by evaluating its encoding $\tau_G(a)$, which takes the~form of~an~expression involving elements of~a~fixed finite generating set $G \subseteq S$;
the~evaluation of~$\varphi \lor \psi$ is done by~nondeterministically evaluating either $\varphi$, or~$\psi$; the~evaluation of~$\varphi \land \psi$ is done by first evaluating~$\varphi$ and~subsequently evaluating~$\psi$ (correctness follows by~distributivity of~$S$).\goodbreak    

The nontrivial part of the proof is to show that $\mathsf{SAT}[S]$ is $\mathbf{NP}[S]$-hard. 
The general idea is to observe that the~usual reduction of a~problem in~$\mathbf{NP}$ to~$\mathsf{SAT}$ can be done such that
even if the formula constructed there is interpreted over $S$, its value is still $0$ or~$1$, based on~its truth value over $\mathbb{B}$ -- in other words, it is \emph{unambiguous} in the sense of~\cite{droste2007a,droste2009c}.
Moreover, this formula can be assembled to~contain a~conjunction ``over all computation steps'' of a~machine $\mathcal{M}$ for~the~problem being reduced, while for~the~$i$-th step it is guaranteed that the $(i+1)$-th
configuration is obtained from the $i$-th using some transition $e$. To deal with the weighted case, it is enough to incorporate the weight $\sigma(e)$ of~such a~transition, so that the~value of~an~assignment corresponding to a~computation $\gamma$ is $\sigma(\gamma)$.

Let us now describe the reduction of a given problem $r \in \mathbf{NP}[S]$ to $\mathsf{SAT}[S]$ in more detail.
Suppose that $r \in S\llangle\Sigma^*\rrangle$ for~some alphabet~$\Sigma$. Then there obviously exists a~polynomial-time weighted Turing machine $\mathcal{M} = (Q,\Gamma,\Delta,\sigma,q_0,F,\square,\triangleright)$ 
over $S$ and~$\Sigma$ with \emph{semi-infinite} tape such that $\|\mathcal{M}\| = r$. The~definition of~such machines is analogous to the double-infinite case. A~difference
is that there is an end-marker $\triangleright \in \Gamma$ at~the~leftmost cell such that the head reading $\triangleright$ can write $\triangleright$ only, while it cannot move to the left;
the head cannot replace other symbols by $\triangleright$. Moreover, let us denote
\begin{displaymath}
\Delta' := \{(p,c,p,c,0) \mid p \in Q;~c \in \Gamma;~(p,c,q,d,s) \not\in \Delta \text{ for all $q \in Q$, $d \in \Gamma$, and $s \in \{-1,0,1\}$}\}.
\end{displaymath}
We write $C \rightarrow_e C'$, for configurations $C,C'$ of $\mathcal{M}$ and~$e = (p,c,p,c,0) \in \Delta'$, if the machine $\mathcal{M}$ is in~the~state~$p$ and~reads~$c$ in~its~configuration~$C$,
while at~the~same time $C' = C$. We write $C \hookrightarrow C'$ if and only if $C \rightarrow_e C'$ for some $e \in \Delta \cup \Delta'$. We also write $\sigma(e) = 1$ for all ``pseudo-transitions'' $e \in \Delta'$.

Without loss of generality, assume that $Q$ and $\Gamma$ are disjoint. Each configuration of~$\mathcal{M}$ is then uniquely determined by~a~word $C \in \Gamma^*Q\Gamma^*$ such that $C$
contains precisely one occurrence of $\triangleright$ that precedes all other symbols from $\Gamma$, and~no~occurrence of~$\square$. The unique occurrence of~a~symbol from~$Q$ represents the~state of~the~machine~$\mathcal{M}$, as well as the position of the machine's head 
(which reads either the~next symbol of~$C$, or~the~first blank cell in~case the~state is the~last symbol of~$C$). Moreover, let us assume that $F = \{q_{\mathit{acc}}\}$ for some single accepting state $q_{\mathit{acc}}$.

As $\mathcal{M}$ runs in~polynomial time, there exists a~function~$f\colon\mathbb{N} \to \mathbb{N}$, defined by~$f(n) = c n^k + d$ for~some $c,k,d \in \mathbb{N}$ and~all $n \in \mathbb{N}$,
such that $\mathsf{TIME}(\mathcal{M},w) \leq f(\lvert w\rvert)$ for~all $w \in \Sigma^*$ and~at~the~same time, every configuration $C$ of~$\mathcal{M}$ reachable upon an~input word~$w$ is of~length at~most $f(\lvert w\rvert)$.
This makes it possible to~``normalise'' the~length of~configurations of~$\mathcal{M}$ by~padding them with blank symbols from the~right -- in~this way, a~configuration $C$ becomes
$C\square^{f(\lvert w\rvert) - \lvert C\rvert}$. \emph{All configurations are understood to be padded like this in~what follows.} 

Let $n := \lvert w\rvert$. It is clear that there is a one-one correspondence between accepting computations in 
\begin{displaymath}
A_w(\mathcal{M}) := \{\gamma \in A(\mathcal{M}) \mid \lambda(\gamma) = w\}
\end{displaymath}
and words $C_0 C_1 \ldots C_{f(n)}$ such that $C_0,\ldots,C_{f(n)}$ are
(padded) configurations of $\mathcal{M}$ such that:
\begin{enumerate}
\item{The configuration $C_0$ is initial for input $w$, \emph{i.e.}, $C_0 = \triangleright q_0 w \square^{f(n) - n - 2}$.
}
\item{One has $C_{i-1} \hookrightarrow C_i$ for $i = 1,\ldots,f(n)$. This is clearly equivalent to saying that 
\begin{displaymath}
C_0 \rightarrow_{e_1} C_1 \rightarrow_{e_2} C_2 \rightarrow_{e_3} \ldots \rightarrow_{e_{f(n)}} C_{f(n)},
\end{displaymath}
where $e_1,\ldots,e_m \in \Delta$ and $e_{m+1} = \ldots = e_{f(n)} \in \Delta'$ for some $m \in \mathbb{N}$. 
}
\item{The configuration $C_{f(n)}$ is accepting, \emph{i.e.}, $C_{f(n)}$ contains $q_{\mathit{acc}}$. 
}
\end{enumerate} 

We now describe a deterministic polynomial-time construction of~a~weighted propositional formula~$\varphi$ for~each $w \in \Sigma^*$ such that \emph{some} of~the (restricted) truth assignments $W\colon X_{\varphi} \to \{0,1\}$ correspond to a~word from $(\Gamma \cup Q)^{(f(n)+1)f(n)}$ representing
an~accepting computation of~$\mathcal{M}$ on~$w$ -- that is, to a~word $C_0 C_1 \ldots C_{f(n)}$ satisfying the three conditions above. In~case this happens for~$\gamma \in A_w(\mathcal{M})$, then $\overline{W}(\varphi) = \sigma(\gamma)$.
If $W$ does not correspond to a~word representing a~computation from $A_w(\mathcal{M})$ (or to a word from $(\Gamma \cup Q)^{(f(n)+1)f(n)}$ at all), then $\overline{W}(\varphi) = 0$. 

The construction of $\varphi$ is done similarly as for nondeterministic Turing machines without weights; however, some additional care is needed in~the~weighted setting.
\emph{The~construction described below finishes the~proof of~the~theorem.} 

The~formula $\varphi$ constructed in~what follows contains variables $x_{i,j,c}$ for~\mbox{$i = 0,\ldots,f(n)$}, \mbox{$j = 1,\ldots,f(n)$}, and~$c \in \Gamma \cup Q$.
The~intuitive meaning of~a~variable $x_{i,j,c}$ is that $W(x_{i,j,c}) = 1$ if and only if the $(if(n) + j)$-th symbol of~the~word in~$(\Gamma \cup Q)^{(f(n)+1)f(n)}$ corresponding to~the assignment~$W$ is~$c$.
In~case the~corresponding word represents a~computation of~$\mathcal{M}$ in~the~sense explained above, then this is precisely the~$j$-th symbol of~the~configuration~$C_i$. 
Of~course, assignments $W\colon X_{\varphi} \to \{0,1\}$ such that $W(x_{i,j,c}) = W(x_{i,j,d}) = 1$ for~some $i,j$ and~two distinct $c,d \in \Gamma \cup Q$ do not correspond to any~word in~$(\Gamma \cup Q)^*$.
Our~construction of~$\varphi$ guarantees that $\overline{W}(\varphi) = 0$ whenever this happens.

At the highest level, we take 
\begin{displaymath}
\varphi := \varphi_{\mathit{valid}} \land \varphi_{\mathit{init}} \land \varphi_1 \land \ldots \land \varphi_{f(n)} \land \varphi_{\mathit{fin}},
\end{displaymath}
where:
\begin{enumerate}
\setcounter{enumi}{-1}
\item{The formula $\varphi_{\mathit{valid}}$ satisfies $\overline{W}(\varphi_{\mathit{valid}}) = 1$ if and only if $W(x_{i,j,c}) = 1$ for \emph{precisely one} $c \in \Gamma \cup Q$
for~\mbox{$i = 0,\ldots,f(n)$} and~$j = 1,\ldots,f(n)$ -- or, in~other words, if and only if $W$ corresponds to a~word in~$(\Gamma \cup Q)^{(f(n)+1)f(n)}$. Otherwise $\overline{W}(\varphi_{\mathit{valid}}) = 0$.  
}
\item{Provided $\overline{W}(\varphi_{\mathit{valid}}) = 1$, the formula $\varphi_{\mathit{init}}$ satisfies $\overline{W}(\varphi_{\mathit{init}}) = 1$ if and only if the~prefix of~length~$f(n)$ of~the~word corresponding to~$W$ is precisely the~initial configuration $C_0 = \triangleright q_0 w \square^{f(n) - n - 2}$;
otherwise $\overline{W}(\varphi_{\mathit{init}}) = 0$.   
}
\item{Let $i \in \{1,\ldots,f(n)\}$. Then, provided $\overline{W}(\varphi_{\mathit{valid}}) = 1$ and~if the~word corresponding to~$W$ starts with~$i$ valid (padded) configurations $C_0,\ldots,C_{i-1}$,
the~formula $\varphi_i$ satisfies $\overline{W}(\varphi_i) = \sigma(e)$ whenever the~next~$f(n)$ symbols of~the~corresponding word form a~configuration $C_i$ such that $C_{i-1} \rightarrow_e C_i$ for~$e \in \Delta \cup \Delta'$;
otherwise $\overline{W}(\varphi_i) = 0$.       
}
\item{Provided $\overline{W}(\varphi_{\mathit{valid}}) = 1$ and in case the~word corresponding to~$W$ consists of $f(n)+1$ valid configurations $C_0,\ldots,C_{f(n)}$, the~formula $\varphi_{\mathit{fin}}$ satisfies $\overline{W}(\varphi_{\mathit{fin}}) = 1$ if and only if $C_{f(n)}$ is accepting (\emph{i.e.}, if it contains $q_{\mathit{acc}}$);
otherwise $\overline{W}(\varphi_{\mathit{fin}}) = 0$.  
}
\end{enumerate}

It is clear that if $\varphi$ is constructed for $w$ in this way, then $(\|\mathcal{M}\|,w) = (\mathsf{SAT}[S],\langle\varphi\rangle)$.
It thus suffices to construct the~formulae $\varphi_{\mathit{valid}},\varphi_{\mathit{init}},\varphi_{1},\ldots,\varphi_{f(n)},\varphi_{\mathit{fin}}$ with the properties above.
We describe these constructions using the ``big'' operators
\begin{displaymath}
\bigvee \qquad \text{ and } \qquad \bigwedge
\end{displaymath}
for convenience. However, note that the use of ``big'' conjunction is problematic, as semi\-ring multiplication might not be commutative. 
We nevertheless use this operator just in the following two contexts:
\begin{enumerate}
\item{We may use it freely when its operands evaluate to $0$ or $1$ under all truth assignments -- as $0$ and $1$ commute, the use of the problematic operator is justified in this case.
}
\item{We also use this operator in some cases when its operands are not guaranteed to evaluate to $0$ or $1$. More precisely, let us fix the~following linear order on~the~subformulae to be constructed:
\begin{displaymath}
\varphi_{\mathit{valid}} \prec \varphi_{\mathit{init}} \prec \varphi_1 \prec \ldots \prec \varphi_{f(n)} \prec \varphi_{\mathit{fin}}.
\end{displaymath}
If $\psi$ denotes any of these subformulae, then we may use the ``big'' conjunction operator in~the~construction of~$\psi$
in~case its operands evaluate to~$0$~or~$1$ under all truth assignments such that subformulae $\psi'$ with $\psi' \prec \psi$ evaluate to a~nonzero value. This means
that the~formula~$\varphi$ is guaranteed to~evaluate to~$0$ whenever some of~the~operands of~a~``big'' conjunction in~$\psi$ evaluate to a~value other than~$0$ or~$1$.       
}  
\end{enumerate}

Clearly, $\varphi_{\mathit{valid}}$ can be constructed as 
\begin{displaymath}
\varphi_{\mathit{valid}} := \bigwedge_{i=0}^{f(n)} \bigwedge_{j = 1}^{f(n)} \bigvee_{c \in \Gamma \cup Q}\left(x_{i,j,c} \land \bigwedge_{\substack{d \in \Gamma \cup Q \\ d \neq c}} \lnot x_{i,j,d}\right).  
\end{displaymath} 
As the content of the parentheses can evaluate to $1$ for at most one $c$ for any given $i$ and $j$, while for all other~$c$ it evaluates to~$0$, the~formula $\varphi_{\mathit{valid}}$ is well-defined and always evaluates either to $0$, or to $1$, according to what has been claimed above. 

In case $w = c_1 \ldots c_n$ for $c_1,\ldots,c_n \in \Sigma$, let $\varphi_{\mathit{init}}$ be defined by
\begin{displaymath}
\varphi_{\mathit{init}} := x_{0,1,\triangleright} \land x_{0,2,q_0} \land x_{0,3,c_1} \land x_{0,4,c_2} \land \ldots \land x_{0,n+2,c_n} \land x_{0,n+3,\square} \land \ldots \land x_{0,f(n),\square}. 
\end{displaymath}
It is clear that $\varphi_{\mathit{init}}$ has the desired properties. 

We now describe the construction of the formula $\varphi_i$ for $i = 1,\ldots,f(n)$. Let
\begin{displaymath}
\varphi_i := \bigvee_{\substack{e \in \Delta \cup \Delta' \\ e = (p,c,q,d,s)}} \sigma(e) \land \bigwedge_{j = 1}^{f(n)} \left(\psi_1(i,j) \lor \psi_2(i,j) \lor \psi_3(i,j,p,c,q,d,s)\right). 
\end{displaymath} 
Here, $\psi_1(i,j)$, $\psi_2(i,j)$, and $\psi_3(i,j,p,c,q,d,s)$ are subformulae to be specified below, satisfying the~following properties in~case $\overline{W}(\varphi_{\mathit{valid}}) = 1$ and~in~case the~word corresponding to~$W$ starts with $i$ valid configurations $C_0,\ldots,C_{i-1}$:
\begin{enumerate}
\item{The formula $\psi_1(i,j)$ satisfies $\overline{W}(\psi_1(i,j)) = 1$ when the~$j$-th symbol of~the~configuration~$C_{i-1}$ -- which we may denote by~$c$ -- cannot be altered in~the~following computation step
(\emph{i.e.}, $c \not\in Q$ and~the~same holds for~both neighbouring symbols of~$c$ in~$C_{i-1}$, in~case they exist) and, at~the~same time, the~$j$-th symbol following the~end of~$C_{i-1}$ is $c$ as well;
otherwise $\overline{W}(\psi_1(i,j)) = 0$. The~symbol~$c$ is thus ``copied to the next configuration''.  
}
\item{The formula $\psi_2(i,j)$ satisfies $\overline{W}(\psi_2(i,j)) = 1$ when the~$j$-th symbol of~the~configuration~$C_{i-1}$ is not in~$Q$, but either the $(j-1)$-th, or the~$(j+1)$-th
symbol of~$C_{i-1}$ is in~$Q$; otherwise $\overline{W}(\psi_2(i,j)) = 0$. This formula just ``skips'' symbols in~$C_{i-1}$ adjacent to~the~symbol representing the~machine's head, as their
correct replacement in~the~following configuration is taken care of by~the~formula $\psi_3(i,j,p,c,q,d,s)$.    
}
\item{The formula $\psi_3(i,j,p,c,q,d,s)$ satisfies $\overline{W}(\psi_3(i,j,p,c,q,d,s)) = 1$ if the $j$-th symbol of~the~configuration $C_{i-1}$ is in~$Q$ and~if the~$(j-1)$-th (if $j - 1 > 0$), the~$j$-th, and~the~$(j+1)$-th symbol
following the~end of~$C_{i-1}$ represent the~symbols obtained from the~respective symbols of~$C_{i-1}$ using the~transition $(p,c,q,d,s)$; otherwise $\overline{W}(\psi_3(i,j,p,c,q,d,s)) = 0$. Note that we may assume $j < f(n)$ here,
as~otherwise the~next ``unpadded'' configuration would be longer than $f(n)$, contradicting the~choice~of~$f$.       
} 
\end{enumerate}   
It is clear that no more than one of the formulae $\psi_1(i,j)$, $\psi_2(i,j)$, $\psi_3(i,j,p,c,q,d,s)$ can evaluate to~$1$ for~given $i,j,p,c,q,d,s$ in~case the~above properties are satisfied. Hence, once the~constructions for~these formulae are described,
the~formula $\varphi_i$ is well-defined and~has the~desired properties.   
   
The formula $\psi_1(i,j)$ can be constructed as follows:
\begin{displaymath}
\psi_1(i,j) := \left(\bigwedge_{q \in Q} \lnot x_{i-1,j-1,q}\right) \land \left(\bigwedge_{q \in Q} \lnot x_{i-1,j,q}\right) \land \left(\bigwedge_{q \in Q} \lnot x_{i-1,j+1,q}\right) \land \left(\bigvee_{c \in \Gamma} x_{i-1,j,c} \land x_{i,j,c}\right) 
\end{displaymath}   
if $0 < j < f(n)$,
\begin{displaymath}
\psi_1(i,j) := \left(\bigwedge_{q \in Q} \lnot x_{i-1,j,q}\right) \land \left(\bigwedge_{q \in Q} \lnot x_{i-1,j+1,q}\right) \land \left(\bigvee_{c \in \Gamma} x_{i-1,j,c} \land x_{i,j,c}\right) 
\end{displaymath}
if $j = 0$, and\goodbreak
\begin{displaymath}
\psi_1(i,j) := \left(\bigwedge_{q \in Q} \lnot x_{i-1,j-1,q}\right) \land \left(\bigwedge_{q \in Q} \lnot x_{i-1,j,q}\right) \land \left(\bigvee_{c \in \Gamma} x_{i-1,j,c} \land x_{i,j,c}\right) 
\end{displaymath}
if $j = f(n)$.

Similarly, the formula $\psi_2(i,j)$ can be constructed as follows:
\begin{displaymath}
\psi_2(i,j) := \left(\bigwedge_{q \in Q} \lnot x_{i-1,j,q}\right) \land \left(\bigvee_{q \in Q} x_{i-1,j-1,q} \lor x_{i-1,j+1,q}\right)
\end{displaymath}  
if $0 < j < f(n)$,
\begin{displaymath}
\psi_2(i,j) := \left(\bigwedge_{q \in Q} \lnot x_{i-1,j,q}\right) \land \left(\bigvee_{q \in Q} x_{i-1,j+1,q}\right)
\end{displaymath}
if $j = 0$, and
\begin{displaymath}
\psi_2(i,j) := \left(\bigwedge_{q \in Q} \lnot x_{i-1,j,q}\right) \land \left(\bigvee_{q \in Q} x_{i-1,j-1,q}\right)
\end{displaymath}
if $j = f(n)$.

Finally, let us construct $\psi_3(i,j,p,c,q,d,s)$. For $0 < j < f(n)$, this should ``rewrite'' a~subword $c'pc$ -- where $c' \in \Gamma$ and $p$ is at~the~$j$-th position of~$C_{i-1}$ -- to some other subword $\alpha(q,s,c')\beta(q,d,s,c')\gamma(q,d,s)$, where
\begin{align*}
\alpha(q,s,c') & := \left\{\begin{array}{ll}q & \text{ if $s = -1$}, \\ c' & \text{ otherwise},\end{array}\right. \\ 
\beta(q,d,s,c') & := \left\{\begin{array}{ll}c' & \text{ if $s = -1$}, \\ q & \text{ if $s = 0$}, \\ d & \text{ if $s = 1$}, \end{array}\right.
\end{align*}
and
\begin{displaymath}
\gamma(q,d,s) := \left\{\begin{array}{ll}q & \text{ if $s = 1$}, \\ d & \text{ otherwise}.\end{array}\right.
\end{displaymath}
The formula $\psi_3(i,j,p,c,q,d,s)$ can thus be constructed as
\begin{displaymath}
\psi_3(i,j,p,c,q,d,s) := \bigvee_{c' \in \Gamma} x_{i-1,j-1,c'} \land x_{i-1,j,p} \land x_{i-1,j+1,c} \land x_{i,j-1,\alpha(q,s,c')} \land x_{i,j,\beta(q,d,s,c')} \land x_{i,j+1,\gamma(q,d,s)} 
\end{displaymath}
in case $0 < j < f(n)$. If $j = 0$, then $s$ cannot be $-1$ and~$\psi_3(i,j,p,c,q,d,s)$ can be constructed as 
\begin{displaymath}
\psi_3(i,j,p,c,q,d,s) := x_{i-1,j,p} \land x_{i-1,j+1,c} \land x_{i,j,\beta(q,d,s,c')} \land x_{i,j+1,\gamma(q,d,s)} 
\end{displaymath}
for any $c' \in \Gamma$. If $j = f(n)$, then we can set, according to what has been observed above,
\begin{displaymath}
\psi_3(i,j,p,c,q,d,s) := 0.
\end{displaymath} 

To complete the proof, it remains to construct the formula $\varphi_{\mathit{fin}}$. However, this can clearly be done as~follows:
\begin{displaymath}
\varphi_{\mathit{fin}} := x_{f(n),1,q_{\mathit{acc}}} \lor x_{f(n),2,q_{\mathit{acc}}} \lor \ldots \lor x_{f(n),f(n),q_{\mathit{acc}}}. 
\end{displaymath}
This finishes the construction.
\end{proof}

Let us finally touch on a possibility of studying $\mathbf{NP}[S]$-completeness with respect to different types of~reductions.
The definitions of the following two reductions -- in our terminology, the \emph{polynomial-time Turing reduction} and the \emph{polynomial-time one-call reduction} --
are in fact based on reductions most typically used to define $\#\mathbf{P}$-completeness \cite{valiant1979a,valiant1979b,arora2009a,papadimitriou1994a}. These reductions are clearly stronger than the~polynomial-time many-one reduction;
this means that the~problems proved to be $\mathbf{NP}[S]$-complete above with respect to $\leq_m$ stay $\mathbf{NP}[S]$-complete under the~reductions introduced below as well.    

Following the usual approach, we define Turing reductions by means of oracle machines.
One usually considers $\mathbf{FP}$-machines with access to~an~oracle $f\colon\Sigma^* \to \mathbb{N}$ for~counting problems, while outputs of~the~oracle are given in~binary.
In~view of~our generalisation of~the~complexity class $\mathbf{FP}$ to $\mathbf{FP}[S]$, it~seems natural to~consider $\mathbf{FP}[S]$-machines with an oracle $f\colon\Sigma^* \to \langle G\rangle$ for some finite $G \subseteq S$,
outputs of~the~oracle being represented as~terms from~$T(G)$.

However, unlike for binary representations of numbers, it is problematic to canonically choose a single term $t \in T(G)$ such that $h_G[S](t) = a$ for given $a \in \langle G\rangle$.
On~the~other hand, allowing the oracle to output any such term is of little use for the machine that accesses it -- for instance, it could start by an exponentially long sum of zeros.
For~these reasons, we define a~\mbox{$G$-oracle} to be a~pair $\mathcal{O} = (f,L)$, where $f\colon\Sigma^* \to \langle G\rangle$ is a~function and~$L \in \mathbf{P}$ is a~``filter'' language;
a~machine with access to~$\mathcal{O}$ can ``ask'' for an~output of~$f$ upon $w \in \Sigma^*$, for which the~oracle ``returns'' a~term $t(w) \in T(G) \cap L$ such that $h_G[S](t(w)) = f(w)$, 
or reaches some special state in~case there is no such $t(w)$.            

\begin{definition}
Let $S$ be a semiring and $\Sigma_1,\Sigma_2$ alphabets. A series $r \in S\llangle\Sigma_1^*\rrangle$ (or~a~function $r\colon\Sigma_1^* \to S$) is \emph{polynomially Turing-reducible} to $s \in S\llangle\Sigma_2^*\rrangle$,
written $r \leq_T s$, if there is a finite set $G \subseteq S$ such that $r \in \langle G\rangle\llangle\Sigma_1^*\rrangle$ and $s \in \langle G\rangle\llangle\Sigma_2^*\rrangle$, and~a~polynomial-time deterministic Turing machine 
with access to~some $G$-oracle $\mathcal{O} = (s,L)$ that transforms each $w \in \Sigma_1^*$ to some term $t(w) \in T(G)$ such that $h_G[S](t(w)) = (r,w)$.     
\end{definition}

By~a~\emph{polynomial-time one-call reduction}, we understand a~polynomial-time Turing reduction, in which the machine can access the oracle at most once. 

As the reductions just introduced are stronger than $\leq_m$ for problems in $\mathbf{NP}[S]$, problems $\mathbf{NP}[S]$-complete with respect to $\leq_m$ remain $\mathbf{NP}[S]$-complete with respect to them as well.
On~the~other hand, the argument ruling out the~existence of~$\mathbf{NP}[S]$-complete problems for semirings that are not finitely generated can no~longer be applied here.  
A~more detailed study of $\mathbf{NP}[S]$-completeness with respect to Turing reductions is left for~future research. 

\section{Hardness of Problems over Different Semirings}

Given semirings $S,S'$, a~semiring homomorphism $\alpha\colon S \to S'$, and a~formal power series $r \in S\llangle\Sigma^*\rrangle$ over some alphabet $\Sigma$,
we denote by $\alpha(r)$ the series 
\begin{displaymath}
\alpha(r) := \sum_{w \in \Sigma^*} \alpha(r,w) w.
\end{displaymath} 
We now prove that the~image of~an~$\mathbf{NP}[S]$-complete series under a~\emph{surjective} homomorphism from~$S$ onto~$S'$ is~$\mathbf{NP}[S']$-complete.

\begin{proposition}
Let $S,S'$ be semirings, $\alpha\colon S \to S'$ a surjective homomorphism, and~$s \in S\llangle\Sigma^*\rrangle$ a~series. If~$s$ is~$\mathbf{NP}[S]$-complete with respect to~$\leq_m$, then
$\alpha(s)$ is $\mathbf{NP}[S']$-complete with respect to~$\leq_m$.
\end{proposition}
\begin{proof}
It is clear that $\alpha(s) \in \mathbf{NP}[S']$ whenever $s \in \mathbf{NP}[S]$. Let $s$ be $\mathbf{NP}[S]$-hard, and~let~us prove that each $r \in \mathbf{NP}[S']$ reduces to~$\alpha(s)$.
Let $r$ be given. As $\alpha$ is surjective, there is $r' \in \mathbf{NP}[S]$ such that $\alpha(r') = r$. Then $r' \leq_m s$. As~a~result, $r = \alpha(r') \leq_m \alpha(s)$.  
\end{proof}

By virtue of~the~first isomorphism theorem, the~proposition established above actually says that if~$s$ is~$\mathbf{NP}[S]$-complete with respect to~$\leq_m$, then its ``natural images''
are $\mathbf{NP}[S']$-complete over the~factor semi\-rings~$S'$ of~$S$. Note that this is not true for stronger reductions -- already for~$S = \mathbb{N}$ and~its factor semiring~$\mathbb{B}$,
the~problem of~bipartite matchings is a~notorious counterexample~\cite{valiant1979a}. 

The following proposition states a~similar property: weighted computation over a~semi\-ring~$S$ is always in some sense ``harder'' than over its factor semirings.   

\begin{proposition}
\label{prop:division}
Let $S$ be a semiring such that $\mathbf{FP}[S] = \mathbf{NP}[S]$. Then $\mathbf{FP}[T] = \mathbf{NP}[T]$ for~all factor semirings~$T$ of~$S$. 
\end{proposition} 
\begin{proof}
When $T$ is a factor semiring of $S$, a surjective homomorphism $\alpha\colon S \to T$ has to exist.
Suppose $\mathbf{FP}[S] = \mathbf{NP}[S]$. Let $r \in T\llangle\Sigma^*\rrangle$ be a formal power series in $\mathbf{NP}[T]$. 
By surjectivity of $\alpha$, there is a~problem $r' \in \mathbf{NP}[S]$ such that $\alpha(r') = r$.
Let $\mathcal{M}$ be an~$\mathbf{FP}[S]$-machine for $r'$. Given $w \in \Sigma^*$, the machine $\mathcal{M}$ outputs a~term $t(w) \in T(G)$, for some finite $G \subseteq S$, such
that $h_G[S](t(w)) = (r',w)$. The $\mathbf{FP}[T]$-machine for~$r$ can~then just simulate $\mathcal{M}$ and finally replace each $g \in G$ by $\alpha(g)$ in $t(w)$; 
let us denote the~term thus obtained by~$\alpha(t(w))$. Clearly $h_{\alpha(G)}[T](\alpha(t(w))) = \alpha(r',w) = (r,w)$.          
\end{proof}

\section{$\mathbf{FP}[S]$~vs.~$\mathbf{NP}[S]$}
\label{sec:fpnp}
 
Recall from Proposition \ref{prop:fpinclnp} that $\mathbf{FP}[S] \subseteq \mathbf{NP}[S]$ for all semirings $S$, generalising both the~inclusion \mbox{$\mathbf{P} \subseteq \mathbf{NP}$} and~the~inclusion $\mathbf{FP} \subseteq \#\mathbf{P}$.
We now give an~example of~$S$, over which provably $\mathbf{FP}[S] \neq \mathbf{NP}[S]$.  

\begin{theorem}
Let $\Sigma = \{a,b,\#\}$. Then $\mathbf{FP}[2_{\mathit{fin}}^{\Sigma^*}] \subsetneq \mathbf{NP}[2_{\mathit{fin}}^{\Sigma^*}]$.
\end{theorem}
\begin{proof}
Let $\mathsf{PAL}$ be a power series in \smash{$2_{\mathit{fin}}^{\Sigma^*}\llangle c^*\rrangle$} defined for all $n \in \mathbb{N}$ by 
\begin{displaymath}
(\mathsf{PAL},c^n) = \{w\#w^R \mid w \in \{a,b\}^n\}.
\end{displaymath}
It is clear that \smash{$\mathsf{PAL} \in \mathbf{NP}[2_{\mathit{fin}}^{\Sigma^*}]$}.

We now prove that $\mathsf{PAL}$ is not in \smash{$\mathbf{FP}_{\Sigma}[2_{\mathit{fin}}^{\Sigma^*}]$} (each $c \in \Sigma$ is identified with $\{c\}$ here), so
that it is not in~\smash{$\mathbf{FP}[2_{\mathit{fin}}^{\Sigma^*}]$} by~Lemma~\ref{le:generators}.
To do so, we show that $\lvert t(n)\rvert \geq 2^n$ for every $t(n) \in T(\Sigma)$ satisfying \smash{$h_{\Sigma}[2_{\mathit{fin}}^{\Sigma^*}](t(n)) = (\mathsf{PAL},c^n)$} and~for~all~$n \in \mathbb{N}$.
 
Let a term $t(n)$ be given and let us form a term $t'(n)$ by~labelling occurrences of~$\#$ in~$t(n)$ uniquely by~$\#_1,\ldots,\#_m$ for some $m \in \mathbb{N}$; fix $\Gamma := \{a,b,\#_1,\ldots,\#_m\}$ and $L := h_{\Gamma}[2_{\mathit{fin}}^{\Gamma^*}](t'(n))$.
Then it is clear that $h(L) = (\mathsf{PAL},c^n)$ for~a~homomorphism $h\colon \Gamma^* \to \Sigma^*$ given by $h(a) = a$, $h(b) = b$, and $h(\#_i) = \#$ for~$i = 1,\ldots,m$.
However, it is also clear that if $u\#_i u^R$ and $v\#_i v^R$ are in $L$ for some $u \neq v$ from $\{a,b\}^n$ and~some $i \in \{1,\ldots,m\}$, then a non-palindrome $u\#_i v^R$ is in $L$ as well.
Hence, each $\#_i$ can correspond to~at~most one $w \in L$, implying $\lvert t(n)\rvert \geq m \geq 2^n$.      
\end{proof}

It would be interesting to know about some other examples of semirings $S$, preferably smaller than $2_{\mathit{fin}}^{\Sigma^*}$ in~the~ordering by~factorisation,
such that provably $\mathbf{FP}[S] \neq \mathbf{NP}[S]$. 

\section*{Acknowledgements}

I would like to thank the anonymous reviewers of all versions of this article for~their helpful suggestions.
Moreover, my special thanks go to Manfred Droste for prompting me to~finally publish this long-forgotten piece of work. 


\bibliographystyle{abbrv}
\bibliography{references}

\end{document}